\documentclass[review]{elsarticle}
\usepackage{lineno,hyperref}
\modulolinenumbers[5]
\biboptions{authoryear}
\usepackage{amsmath,amsfonts,amssymb}
\usepackage{changepage}
\usepackage{graphicx}
\usepackage{url}
\usepackage{rotating}
\usepackage{pdflscape}
\usepackage[utf8]{inputenc}
\usepackage[english]{babel}
\usepackage{paralist, algorithmic,algorithm}
\usepackage{url}
\usepackage[usenames]{color}
\usepackage{placeins}
\usepackage{float}
\usepackage{subfig}
\usepackage{mathabx}
\usepackage{afterpage}
\usepackage[export]{adjustbox}
\usepackage[outdir=./]{epstopdf}
\usepackage{tikz}
\usetikzlibrary{shapes,shapes,arrows,fit,calc,positioning,automata,matrix,shadows}  
\usepackage{xr}
\makeatletter
\newcommand*{\addFileDependency}[1]{
  \typeout{(#1)}
  \@addtofilelist{#1}
  \IfFileExists{#1}{}{\typeout{No file #1.}}
}
\makeatother

\catcode`\^ = 13 \def^#1{\sp{#1}{}}

\definecolor{Pink}{rgb}{1.0, 0.5, 0.5}
\definecolor{Maroon}{rgb}{0.8, 0.0, 0.0}

\def\boxit#1{\vbox{\hrule\hbox{\vrule\kern6pt\vbox{\kern6pt#1\kern6pt}\kern6pt\vrule}\hrule}}

\newcommand{\B}{\mathbf}

\newcommand{\mA}{\mathcal{A}}

\newcommand{\M}{\mbox}

\newtheorem{theorem}{Theorem}[section]
\newtheorem{lemma}[theorem]{Lemma}

\newenvironment{proof}[1][Proof]{\begin{trivlist}
\item[\hskip \labelsep {\bfseries #1}]}{\end{trivlist}}

\newcommand{\bM}{\mbox{\bf M}}

\newcommand{\bu}{\mbox{\bf u}}
\newcommand{\bv}{\mbox{\bf v}}
\newcommand{\bw}{\mbox{\bf w}}

\newcommand{\bB}{\mbox{\bf B}}
\newcommand{\bC}{\mbox{\bf C}}

\newcommand{\balpha}{\mbox{\boldmath $\alpha$}}

\newcommand{\bbeta}{\mbox{\boldmath $\beta$}}

\newcommand{\bGamma}{\mbox{\boldmath $\Gamma$}}

\newcommand{\bmu}{\mbox{\boldmath $\mu$}}

\newcommand{\bPhi}{\mbox{\boldmath $\Phi$}}

\newcommand{\what}{\widehat}

\newcommand{\diag}{\mathrm{diag}}

\def\t{^\top}

\def\beqn{\begin{eqnarray}}
\def\eeqn{\end{eqnarray}}
\def\beqns{\begin{eqnarray*}}
\def\eeqns{\end{eqnarray*}}

\def\0{{\bf 0}}
\def\A{{\bf A}}
\def\a{{\bf a}}
\def\b{{\bf b}}
\def\C{{\bf C}}
\def\c{{\bf c}}
\def\D{{\bf D}}
\def\d{{\bf d}}

\def\E{{\bf E}}

\def\I{{\bf I}}

\def\t{{\bf t}}

\def\bO{{\bf O}}
\def\bP{{\bf P}}

\def\U{{\bf U}}
\def\V{{\bf V}}
\def\S{{\bf S}}

\def\u{{\bf u}}
\def\v{{\bf v}}
\def\W{{\bf W}}
\def\w{{\bf w}}
\def\X{{\bf X}}
\def\x{{\bf x}}

\def\Y{{\bf Y}}
\def\y{{\bf y}}
\def\Z{{\bf Z}}
\def\z{{\bf z}}
\def\1{{\bf 1}}

\def\J{{\bf J}}

\def\trans{^{\rm T}}
\def\strans{^{*\rm T}}
\newcommand{\bbR}{\mathbb{R}}
\newcommand{\td}{\tilde{d}}

\newcommand{\bbL}{\mathcal{L}}
\newcommand{\mcA}{\mathcal{A}}
\newcommand{\gammaL}{\gamma^l}

\DeclareMathOperator*{\argmin}{arg\,min}
\newcommand{\bs}{\boldsymbol}
\DeclareMathOperator{\Tr}{tr}

\begin{document}
\sloppy

\begin{frontmatter}
\title{Generalized Co-sparse Factor Regression\tnoteref{t1}}
\tnotetext[t1]{For this work, there exists supplementary materials providing all the proofs, reproducible simulation code,  additional plots,  tables showing model evaluation, and the application  data to demonstrate model efficacy.}
\author[1]{Aditya Mishra\corref{mycorrespondingauthor}}
\cortext[mycorrespondingauthor]{Corresponding author}
\ead{amishra@flatironinstitute.org}

\author[2]{Dipak K. Dey}
\author[3]{Yong Chen}
\author[2]{Kun Chen}

\address[1]{Center for Computational Mathematics, Flatiron Institute,  New York, NY 10010, USA}
\address[2]{Department of Statistics, University of Connecticut, Storrs, CT 06269, USA}
\address[3]{Division of Biostatistics,  Department of Biostatistics, Epidemiology and Informatics,  University of Pennsylvania Perelman School of Medicine,  Philadelphia, PA 19104, USA}
\begin{abstract}
Multivariate regression techniques are commonly applied to explore the associations between large numbers of outcomes and predictors. In real-world applications, the outcomes are often of mixed types, including continuous measurements, binary indicators, and counts, and the observations may also be incomplete. Building upon the recent advances in mixed-outcome modeling and sparse matrix factorization, 
\textit{generalized co-sparse factor regression} (GOFAR) is proposed, which utilizes the flexible vector generalized linear model framework and encodes the outcome dependency through a sparse singular value decomposition (SSVD) of the integrated natural parameter matrix. 
To avoid the estimation of the notoriously difficult joint SSVD, GOFAR  proposes both \textit{sequential} and \textit{parallel} unit-rank estimation procedures.
By combining the ideas of alternating convex search and majorization-minimization, an efficient algorithm with guaranteed convergence is developed to solve the sparse unit-rank problem and implemented in the R package \textsf{gofar}. Extensive simulation studies and two real-world applications demonstrate the effectiveness of the proposed approach. 
\end{abstract}

\end{frontmatter}
\linenumbers

\section{Introduction}\label{sec2:gecure:intro}

Advances in science and technology have led to exponential growth in the collection of different types of large data in various fields, including healthcare, biology, economics, and finance. 
Many problems of interest amount to exploring the association between multivariate outcomes/responses and multivariate predictors/features. For example, in an ongoing project of the Framingham Heart Study \citep{cupples2007framingham},  researchers are interested in understanding the effect of single nucleotide polymorphisms on multiple phenotypes related to cardiovascular disease. 
Some phenotypes are binary,  depicting medical conditions, whereas others, such as  cholesterol levels, are continuous. In the Longitudinal Study of Aging (LSOA)  \citep{stanziano2010review}, it is of interest to understand the association between future health status (memory, depression, cognitive ability) and predictors such as demographics, past medical conditions, and daily activities. Here some outcome measurements, such as memory score, are continuous, while others are of the categorical/indicator type.

As exemplified by the aforementioned problems, the outcome variables collected in real-world studies are often of mixed types. Moreover, the data can be of large dimensionality and may contain a substantial number of missing values. It is apparent that classical multivariate linear regression (MLR) is no longer applicable, and the approach of separately regressing each response using the predictors via a generalized linear model may also perform poorly because it ignores the potential dependency of the mixed outcomes. Our main objective in this paper is thus to tackle the problem of modeling mixed and incomplete outcomes with large-scale data.

Many existing multivariate regression methods focus on continuous outcomes. Principal component regression \citep{jolliffe1982note} and multivariate ridge regression \citep{Hoerl1970,Brown1980} focus on tackling the problem of multicollinearity among predictors. Reduced-rank regression \citep{anderson1951,reinsel1998,bunea2011} achieves dimension reduction and information sharing by assuming that all the responses are related to a small set of latent factors. Sparse \citep{tib1996} multivariate regression models \citep{turlach2005,peng2010,obo2011} take advantage of certain shared sparsity patterns in the association structure. Regularized multivariate models often boil down to matrix approximation problems; see, e.g., singular-value penalized models \citep{yuan2007, negahban2011estimation, koltch2011, chen2012ann}, and sparse matrix factorization models \citep{chen2012jrssb, chen2012sparse, bunea2012joint, MaSun2014,mishra2017sequential}.

Until recently, only a handful of methods have attempted to solve the modeling challenge with non-Gaussian and mixed outcomes. \citet{cox1992response} and \citet{fitzmaurice1995regression} proposed a likelihood-based approach for bivariate responses in which one variable is discrete and the other is continuous. \citet{prentice1991estimating} and \citet{zhao1992multivariate} utilized the generalized estimating equations framework to obtain mean and covariance estimates. \citet{she2011} and \citet{yee2003} studied the reduced-rank vector generalized linear model (RR-VGLM), assuming the outcomes are of the same type and are from an exponential family distribution \citep{jorgensen1987exponential}. Recently, \citet{chenandluo2017} proposed \emph{mixed-outcome reduced-rank regression} (mRRR), extending the RR-VGLM to the more realistic scenario of mixed and incomplete outcomes. However, the method only considered rank reduction, rendering it inapplicable when many redundant or irrelevant variables are present.

Building upon the recent advances in mixed-outcome modeling and sparse matrix factorization, we propose \textit{generalized co-sparse factor regression}, which utilizes the flexible vector generalized linear model framework \citep{she2011,chenandluo2017} and encodes the outcome dependency through an appealing sparse singular value decomposition (SVD) of the integrated natural parameter matrix. The co-sparse SVD structure in our model, i.e., the fact that both the left and the right singular vectors are sparse, implies a flexible dependency pattern between the outcomes and the predictors: on one hand, the model allows a few latent predictors to be constructed from possibly different subsets of the original predictors, and on the other hand, the model allows the responses to be associated with possibly different subsets of the predictors. The model also covers the generalized matrix completion problem under unsupervised learning. Motivated by \citet{chen2012jrssb} and \citet{mishra2017sequential}, we propose computationally efficient divide-and-conquer procedures to conduct model estimation. The main idea is to extract unit-rank components of the natural parameter matrix in either a \textit{sequential} or a \textit{parallel} way, thus avoiding the difficult joint estimation alternative. 
Each step solves a generalized co-sparse unit-rank estimation problem, and these problems differ only in their offset terms, which are designed to account for the effects of other non-targeted unit-rank components. Our model also allows us to deal with the missing values in the same way as in the celebrated matrix completion. To the best of our knowledge, our approach is among the first to enable both variable selection and latent factor modeling in analyzing incomplete and mixed outcomes.

The rest of the paper is organized as follows. We propose a generalized co-sparse factor regression model in Section \ref{sec2:method}. Section \ref{sec:dac} proposes divide-and-conquer estimation procedures, which reduce the problem to a set of generalized unit-rank estimation problems; these are then studied in detail in Section \ref{sec2:computation}. We  study the  large sample property of the estimator in a unit step in Section \ref{sec2:theory}. 
Section \ref{sec2:simulation} shows the effectiveness of the proposed procedures via extensive simulation studies. Two applications, one on the longitudinal study of aging and the other on sound annotation, are presented in Section \ref{sec:gsecure-application}. We provide some concluding remarks in Section \ref{sec:dis}. All the proofs are provided in Supplementary Materials.

\section{Generalized Co-Sparse Factor Regression}\label{sec2:method}
Consider the multivariate regression setup with $n$ instances of independent observations, forming a response/outcome matrix $\Y = [y_{ik}]_{n \times q}   = [\y_1,\ldots,\y_n]\trans \in\mathbb{R}^{n\times q}$, a predictor/feature matrix $\X =[\x_1,\ldots,\x_n]\trans \in \mathbb{R}^{n\times p}$, and a control variable matrix $\Z =  [\z_1,\ldots,\z_n]\trans \in \mathbb{R}^{n\times p_z}$. 
$\Z$ consists of a set of variables that should always be included in the model and are thus not regularized. Depending on the application, we consider experimental input such as age or gender (factor variable) as control variables.

 We assume that each of the response  variables follows a distribution in the exponential-dispersion family \citep{jorgensen1987exponential}. The probability density function of the $i$th entry in the $k$th outcome, $y_{ik}$, is given by 
\begin{align}
f(y_{ik};\theta_{ik}^*,\phi_k^*) = \exp\left\{\frac{y_{ik}\theta_{ik}^*-b_k(\theta_{ik}^*)}{a_k(\phi_k^*)}+c_k(y_{ik};\phi_k^*)\right\},\label{eq:glmmodel}
\end{align}
where $\theta_{ik}^*$ is the natural parameter, $\phi_k^* \in \mathbb{R}^+$ is the dispersion parameter, and $\{a_k(\cdot)$, $b_k(\cdot)$, $c_k(\cdot)\}$ are functions determined by the specific distribution; see Table 1 in  Supplementary Materials for more details on some of the standard distributions in the exponential family,  e.g., Gaussian, Poisson and Bernoulli. 
We collectively denote the natural parameters of $\Y$  by $\bs\Theta^* = [\theta_{ik}^*]_{n \times q}  \in \mathbb{R}^{n \times q}$ and the  dispersion parameters by $\bPhi^* = \diag[a_1(\phi_1^*), \ldots, a_q(\phi_q^*)]$. Let $g_k=(b_k')^{-1}$ be the canonical link function. Consequently,  $\mathbb{E}(y_{ik})= b_k'(\theta_{ik}^*)= g_k^{-1}(\theta_{ik}^*)$, where $b_k^'(\cdot)$ denotes the derivative function of $b_k(\cdot)$.

We model the natural parameter matrix $\bs\Theta^*$ as
\begin{align}
\bs \Theta(\C^*,\bbeta^*, \bO) = \bO + \Z\bbeta^* + \X\C^*, \label{eq:ThetaDef}
\end{align}
where $ \bO  = [o_{ik}]_{n \times q} \in \mathbb{R}^{n \times q}$ is a fixed offset term, $\C^*=[\c_1^*,\ldots,\c_q^*] \in \mathbb{R}^{p\times q}$ is the coefficient matrix corresponding to the predictors, and $\bbeta=[\bbeta_1^*,\ldots,\bbeta_q^*]  \in \mathbb{R}^{p_z\times q}$ is the coefficient matrix corresponding to the control variables. 
The intercept is included by taking the first column of $\Z$ to be $\1_n$, the $n\times 1$ vector of ones. For simplicity, we may write $\bs\Theta(\C^*,\bbeta^*, \bO)$ as $\bs\Theta^*$ if no confusion arises.

To proceed further, we define some notations. The $k$th column of $\bs\Theta^*$ is denoted  $\bs\Theta_{.k}^*$, and consequently $\b_k(\bs\Theta_{.k}^*) = [b_k(\theta_{ik}^*), \ldots, b_k(\theta_{nk}^*)]\trans$. The element-wise derivative vector of $\b_k(\bs\Theta_{.k}^*)$ is $\b_k^'(\bs\Theta_{.k}^*) = [b_k^'(\theta_{ik}^*), \ldots,b_k^' (\theta_{nk}^*)]\trans$. We then define
\begin{align}
\bB(\bs\Theta^*) = [\b_1(\bs\Theta_{.1}^*),\ldots,\b_q(\bs\Theta_{.q}^*)], \quad \bB^'(\bs\Theta^*) = [\b_1^'(\bs\Theta_{.1}^*),\ldots,\b_q^'(\bs\Theta_{.q}^*) ]. \label{eq:defbtheta}
\end{align}
Similarly, $\bB''(\bs\Theta^*)$ denotes the second-order derivative of $\bB(\bs\Theta^*)$.

We assume the outcomes are conditionally independent given $\X$ and $\Z$. Then the joint negative log-likelihood function is given by 
\begin{align}
\bbL(\bs\Theta^*, \bPhi^*) =  -\sum_{i=1}^n\sum_{k=1}^q \ell_k(\theta_{ik}^*,\phi_k^*)  ,\label{eq:negLogLi}
\end{align}
where $\ell_k(\theta_{ik}^*,\phi_k^*) = \log{f(y_{ik};\theta_{ik}^*,\phi_k^*)}$. Using the definition from \eqref{eq:defbtheta}, a convenient  representation of \eqref{eq:negLogLi} is given by 
\begin{align}
\bbL(\bs\Theta^*,\bPhi^*)  = - \Tr(\Y\trans \bs\Theta^* \bPhi^{-1}^* ) + \Tr({\J}\trans\bB(\bs\Theta^*)\bPhi^{-1}^*),\label{eq:negLmat}
\end{align}
where $\J = \1_{n \times q}$ and $ \Tr( \A)$ is the \textit{trace} of a square  matrix $\A$. In the presence of missing entries in $\Y$, let us define an index set of the observed outcomes as
$$
\bs\Omega = \{(i,k); y_{ik} \mbox{ is observed}, i=1,\ldots,n, k=1,\ldots,q\},
$$
and denote the projection of $\Y$ onto $\bs\Omega$ by $\widetilde{\Y}= \mathcal{P}_{\bs\Omega}(\Y)$, where $\tilde{y}_{ik} = y_{ik}$ for any $(i,k)\in\bs\Omega$ and  $\tilde{y}_{ik}=0$ otherwise. 
Accordingly, the negative log-likelihood function with incomplete data is given by 
\begin{align*}
\bbL(\bs\Theta^*,\bPhi^*)  = - \Tr( \widetilde{\Y}\trans\bs\Theta^* \bPhi^{-1}^* ) + \Tr( \widetilde{\J}\trans\bB(\bs\Theta^*)\bPhi^{-1}^* ) , 
\end{align*}
where $\widetilde{\J} = \mathcal{P}_{\bs\Omega}(\J)$. 
Henceforth, we mainly focus on the complete data case \eqref{eq:negLmat} when presenting our proposed model, as the extension to the missing data case by and large only requires replacing $\Y$ by $\widetilde{\Y}$ and $\J$ by $\widetilde{\J}$.

Without imposing additional structural assumptions on the parameters, maximum likelihood estimation, i.e.,  minimizing $\bbL(\bs\Theta,\bPhi)$ with respect to  $\{\C, \bbeta, \bPhi\}$ for $\bs\Theta = \bO + \X\C + \Z\bbeta$, does not work in high-dimensional settings. The marginal modeling approach, i.e., the fitting of a univariate generalized linear model (uGLM) (or its regularized version) for each individual response, would ignore the dependency among the outcomes. The mixed reduced rank regression (mRRR) \citep{chenandluo2017} imposes a rank constraint on $\C$, but its usage is limited as it does not explore variable selection.

We assume that the regression association is driven by a few latent factors, each of which is constructed from a possibly different subset of the predictors, and, moreover, that each response may be associated with a possibly different subset of the latent factors. To be specific, this amounts to assuming a \emph{co-sparse} SVD of $\C^*$ \citep{mishra2017sequential}, i.e., we decompose $\C^*$ as 
\begin{align}
\C^* = \U^* \D^*\V^*\trans , \qquad \M{s.t.} \qquad \U^*\trans\X\trans\X\U^*/n = \V^*\trans\V^* = \I_{r^*},\label{sec2:intro:cdef}
\end{align}
where both the left singular vector matrix $\U^* = [\u_1^*,\ldots,\u^*_{r^*}] \in \bbR^{p \times {r^*}}$ and the right singular vector matrix $\V^* = [\v_1^*,\ldots,\v^*_{r^*}] \in \bbR^{q \times r}$ are assumed to be \emph{sparse}, and $\D = \M{diag}\{d_1^*,\ldots,d_{r^*} \} \in \bbR^{{r^*} \times {r^*}}$ is the diagonal matrix with the nonzero singular values on its diagonal. 
 The orthogonality constraints ensuring identifiability suggest that the sample latent factors, i.e., $(1/\sqrt{n})\X\u_k^*$ for $k=1,\ldots,r^*$, are uncorrelated with each other, and the strength of the association between the latent factors and the multivariate response $\Y$ is denoted by the singular values $\{d_1^*,\ldots,d_{r^*} \}$.
Figure \ref{fig:gofar-model} shows a diagram of the proposed model structure. 
We thus term the proposed model  {\bf G}eneralized c{\bf o}-sparse {\bf fa}ctor {\bf r}egression (GOFAR).\\

\tikzstyle{xu}=[draw, fill=blue!20, text width=2.5em, 
    text centered, minimum height=7em]

\tikzstyle{xv}=[draw, fill=red!20, text width=6em, 
    text centered, minimum height=2.5em,drop shadow]

\tikzstyle{ann} = [text width=5em, text centered]

\tikzstyle{by} = [xu, text width=5em, fill=red!20, 
    minimum height=7em, rounded corners, drop shadow]

\tikzstyle{bx} = [xu, text width=6em, fill=blue!20, 
    minimum height=10em, rounded corners, drop shadow]
    
\tikzstyle{bg} = [xu, text width=4em, fill=gray!20, 
    minimum height=7em, rounded corners, drop shadow]
    
\tikzstyle{bg1} = [xu, text width=5em, fill=gray!20, 
    minimum height=7em, rounded corners, drop shadow]
    
\tikzstyle{bxx} = [xu, text width=10em, fill=gray!20, 
    minimum height=7em, rounded corners, drop shadow]
    
\tikzstyle{bxd} = [xu, text width=1em, fill=gray!30, 
    minimum height=1em, rounded corners, drop shadow]
    
\tikzstyle{buu} = [xu, text width=1em, fill=white!20, 
    minimum height=10em, rounded corners, drop shadow]

\tikzstyle{bvv}=[xu, text width=4em, fill=white!20, 
    minimum height=1em, rounded corners, drop shadow]
    
\tikzstyle{spsv} = [xu, text width=1em, fill=gray!30, 
    minimum height=1em]
    
\tikzstyle{spsv1} = [xu, text width=0.1em, fill=gray!30, 
    minimum height=1.5em]

\tikzstyle{myarrows}=[line width=0.3mm,draw=black,-triangle 45,postaction={draw, line width=1.5mm, shorten >=4mm, -}]

\usetikzlibrary{arrows, decorations.markings}
\tikzstyle{vecArrow} = [thick, decoration={markings,mark=at position
   1 with {\arrow[semithick]{open triangle 60}}},
   double distance=1.4pt, shorten >= 5.5pt,
   preaction = {decorate},
   postaction = {draw,line width=1.4pt, white,shorten >= 4.5pt}]
\tikzstyle{innerWhite} = [semithick, white, line width=1.4pt, shorten >= 4.5pt]

\begin{figure}
\centering
\resizebox{0.8\textwidth}{!}{
\begin{tikzpicture}[remember picture]
    \node at (2,-2)  (by) [by, label = below:Response]{$\textbf{Y}$};
    \path (by.east)+(2.5,0) node (th) [bg1, label = below:Natural parameter] {$\boldsymbol\Theta^*$};
    \path (th.east)+ (0.5,0) node (sgeq) [ann] {$=$};
    \path (sgeq.east)+(0.5,0) node (of) [bg1,label = below:Offset] {$\textbf{O}$};
     \path (of.east)+ (0.5,0) node (sgp) [ann] {$+$};
        \path (sgp.east)+(0.5,0) node (xc1) [bg1, label = below:Control effect ] {$\textbf{Z}\bbeta^*$};
        \path (xc1.east)+ (0.5,0) node (sgp1) [ann] {$+$};
                \path (sgp1.east)+(0.5,0) node (xc2) [bg1, label = below:Predictor effect ] {$ \textbf{XC}^*$};
    \path (by.south)+(0,-2.5) node (xc) [bg1, label = below:Predictor effect] {$\textbf{XC}^*$};
       \path (xc.east)+ (0.3,0) node (sgp1) [ann] {$=$};
     \path (sgp1.east)+(1.2,0) node (xpr) [bxx, label = below:Predictors] {$\textbf{X}$}; 
     \path (xpr.east)+(0.6,-0.5) node (u1) [buu] {$\textbf{u}_1^*$}; 
     \path (xpr.east)+(0.6,1.0) node (m) [spsv]{}; 
      \path (xpr.east)+(0.6,0.3) node (m) [spsv]{}; 
       \path (xpr.east)+(0.6,-1.3) node (m) [spsv]{}; 
     \path (u1.east)+(0.5,1.5) node (d1) [bxd] {$d_1^*$}; 
     \path (d1.east)+(1.1,0) node (v1) [bvv] {$\textbf{v}_1^*$}; 
        \path (d1.east)+(.5,0) node (m) [spsv1]{}; 
      \path (d1.east)+(1.6,0) node (m) [spsv1]{};

     \path (u1.east)+ (3.7,0) node (dots1) [ann] {$+\,\,\cdots\cdots\,\,+$};
     
     \path (xpr.east)+(6.7,-0.5) node (ur) [buu] {$\textbf{u}_{r^*}^*$}; 
          \path (xpr.east)+(6.7,-1.9) node (m) [spsv]{}; 
      \path (xpr.east)+(6.7,1.0) node (m) [spsv]{}; 
     \path (ur.east)+(0.5,1.5) node (dr) [bxd] {$d_{r^*}^*$}; 
     \path (dr.east)+(1.1,0) node (vr) [bvv] {$\textbf{v}_{r^*}^*$}; 
     
      \path (dr.east)+(1.7,0) node (m) [spsv1]{};

     \node [draw=blue, fit= (u1) (vr),inner sep=0.15cm, line width=1.6pt, label=below:Sparse factorization of $\textbf{C}^*$] {};

     \draw [vecArrow](th.west)+(-0.1,0)-- node[below] {GLM} (by.east);
\end{tikzpicture}
}
\caption{GOFAR: Generalized co-sparse factor regression, modeling a multivariate mixed response matrix $\Y$ using a predictor matrix $\X$ with sparse singular vector components of the low-rank coefficient matrix $\C^*$. } \label{fig:gofar-model}
\end{figure}
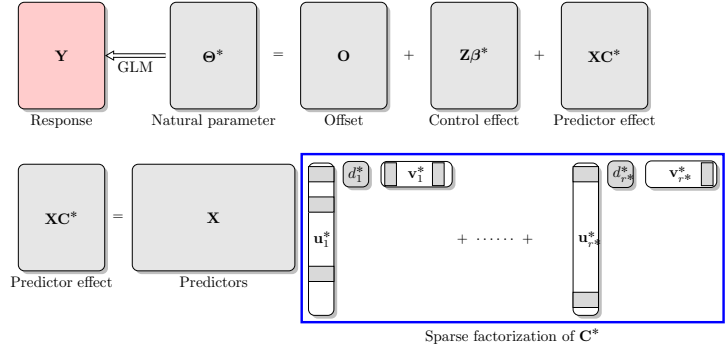

\section{Divide-and-Conquer Estimation Procedures}\label{sec:dac}

Rather than jointly estimating all the sparse singular vectors simultaneously, which may necessarily involve identifiability constraints such as orthogonality in optimization \citep{Uematsu2019}, we take a divide-and-conquer approach. The main idea is to extract the unit-rank components of $\X\C$ one by one, in either a sequential or a parallel way. In this way, we are able to divide the main task into a set of simpler unit-rank problems, which we then conquer in Section \ref{sec2:computation}.

\subsection{Sequential Approach}

Motivated by \citet{mishra2017sequential}, we propose to sequentially extract the unit-rank components of $\C$, i.e., $(d_k,\u_k,\v_k)$, for $k=1,\ldots, r$.
The resulting method is termed  generalized co-sparse factor regression
via sequential extraction (GOFAR(S)).

Algorithm \ref{sec2:alg:1} and Figure \ref{fig:gofarsp}  summarize the computation procedure. In step $k=1$, we conduct the following generalized co-sparse unit-rank estimation (G-CURE),
\begin{align}
(\hat{d}_1, \what{\u}_1, \what{\v}_1, \what{\bbeta},\what{\bPhi}) & \equiv \argmin_{\u,\d,\v, \bbeta, \bPhi}\,\,\,  \bbL(\bs\Theta,\bPhi) +  \rho(\C;\lambda),\label{eq:prock1} \\ &\M{s.t.} \quad \C= d\u\v\trans, \, \u\trans\X\trans\X\u/n = \v\trans\v = 1, \bs\Theta = \bs\Theta(\C, \bbeta, \bO^{(1)}), \notag 
\end{align}
where $\bO^{(1)} = \bO$ (the original offset matrix), and $\rho(\C;\lambda)$ is a sparsity-inducing penalty function with tuning parameter $\lambda$. {
We discuss the formulation of $\rho(\C;\lambda)$ in Section \ref{subsec:penfun} and the selection of tuning parameter $\lambda$ in Section     \ref{sec2:TUNING}}. To streamline the presentation, for now let us assume that we are able to solve G-CURE and select the tuning parameter $\lambda$ suitably. Denote the produced unit-rank solution of $\C$ as $\what{\C}_1 = \what{d}_1\what{\u}_1\what{\v}_1\trans$.

In the subsequent steps, i.e., for $k=2,\ldots r$, we repeat G-CURE each time with an updated offset term,
\begin{align}
	\bO^{(k)} = \bO+ \X\sum_{i=2}^k\what{\C}_{i-1}. \label{eq:gcure:ofset}	
\end{align}
In general, the G-CURE problem in the $k$th step, for $k=1,\ldots, r$, can be expressed as 
\begin{align}
(\hat{d}_k, \what{\u}_k, \what{\v}_k, \what{\bbeta},\what{\bPhi}) & \equiv  \argmin_{\u,\d,\v, \bbeta, \bPhi}\,\,\,  \bbL(\bs\Theta,\bPhi) +  \rho(\C; \lambda),\label{eq:prockk} \\ &\M{s.t.} \quad \C= d\u\v\trans, \, \u\trans\X\trans\X\u/n = \v\trans\v = 1, \bs\Theta = \bs\Theta(\C, \bbeta, \bO^{(k)} ). \notag 
\end{align}
We remark that the low-dimensional parameters $\bbeta$ and $\bPhi$ are re-estimated at the intermediate steps,  and their final estimates are obtained from the last step.

The rationale of the proposed procedure can be traced back to the power method for computing SVD. In each step, through the construction of the offset term, the regression effects from the previous steps are adjusted or ``deflated'' in order to enable G-CURE  to target  a new unit-rank component. The procedure terminates after a pre-specified number of steps or when $\what{d}_k$ is estimated to be zero.

\begin{figure}[htp]
 \centering
  \includegraphics[width=1.0\textwidth]{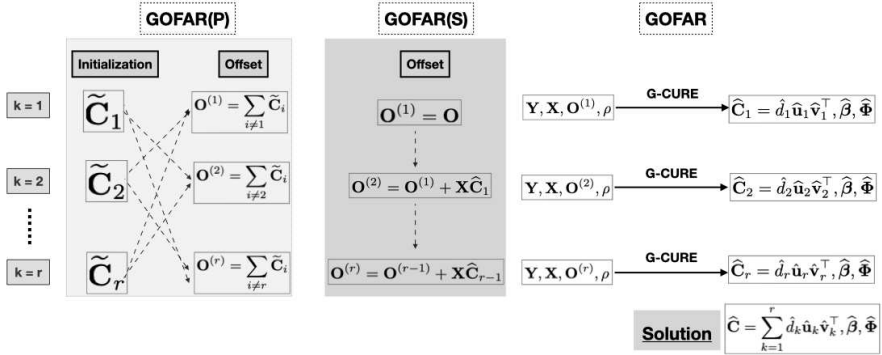}
		\caption{Estimation procedure for the generalized co-sparse factor regression (GOFAR)
           via sequential (GOFAR(S)) and parallel (GOFAR(P)) extraction.} 
     \label{fig:gofarsp}%
\end{figure}

\begin{algorithm}[htp]
	\caption{Generalized Co-Sparse Factor Regression via Sequential Extraction}
	\begin{algorithmic}\label{sec2:alg:1}
          \STATE Initialize: $\bbeta^{(0)}$, $\bPhi^{(0)}$, and set the maximum number of steps $r \geq 1$, e.g., an upper bound of $\M{rank}(\C)$.
		\FOR{$k \gets 1 \, \M{ to }\, r$}
		\STATE (1) Update offset: $\bO^{(k)} = \bO+ \X\sum_{i=2}^k\what{\C}_{i-1}$
		\STATE (2) G-CURE with tuning (see Section \ref{sec2:computation}):\\
                       $(\what{d}_k,\what{\u}_k,\what{\v}_k,\what{\bbeta},\what{\bPhi})  =  \mbox{G-CURE}(\C,\bbeta,\bPhi; \Y, \X,\bO^{(k)}, \rho)$, and $\what{\C}_k  =  \what{d}_k\what{\u}_k\what{\v}_k\trans$.                
		\IF{$\what{d}_k=0$} 
		\STATE{Set $\what{r} = k$; $k \leftarrow r$;} 
		\ENDIF
		\ENDFOR
				\RETURN{$\what{\C} = \sum_{k=1}^{\what{r}} \what{\C}_k$,  $\what{\bbeta}$, $\what{\bPhi}$.}
	\end{algorithmic}
\end{algorithm}

\subsection{Parallel extraction}\label{sssec:eea}

When the true rank of $\C$ is moderate or high, the above sequential extraction procedure may be time consuming. This motivates us to also consider generalized co-sparse factor regression via parallel extraction (GOFAR(P)), in which the construction of the offset terms for targeting different unit-rank components is based on some computationally efficient initial estimator of $\C$.

We summarize the GOFAR(P) procedure in Algorithm \ref{alg:geecure} and Figure \ref{fig:gofarsp}. Given a desired rank $r$, we first obtain an initial estimate of $\C$, denoted  $\widetilde{\C}$, by solving an initialization problem denoted  $ \M{G-INIT}(\C,\bbeta,\bPhi; \Y, \X,\bO,r)$; see Section 1.3 of Supplementary Materials for more details. The initial estimates of the unit-rank components are then computed from the SVD of the regression components $\X\widetilde{\C}$, 
$$
\widetilde{\C} = \sum_{k=1}^r \widetilde{\C}_k= \widetilde{\U}\widetilde{\D}\widetilde{\V}\trans, \qquad \M{s.t. } \widetilde{\U}\trans \X\trans\X \widetilde{\U}/n = \widetilde{\V}\trans\widetilde{\V} = \I_r,
$$
where $\widetilde{\U} = [\widetilde{\u}_1, \ldots, \widetilde{\u}_r] \in \mathbb{R}^{p \times r}$, $\widetilde{\V} = [\widetilde{\v}_1, \ldots, \widetilde{\v}_r] \in \mathbb{R}^{q \times r}$, $\widetilde{\D} = \M{diag}[\widetilde{d}_1, \ldots, \widetilde{d}_r] \in \mathbb{R}^{r \times r}$, and  $\widetilde{\C}_k =  \widetilde{d}_k\widetilde{\u}_k\widetilde{\v}_k\trans$.

The required offset terms for targeting different components are computed based on $\widetilde{\C}$ as
\begin{align}
\widetilde{\bO}^{(k)} =\bO + \X \sum_{i \neq k}\widetilde{\C}_i, \qquad k =1,\ldots, r. \label{eq:offset-gofarp}
\end{align}
Then, the problems $\mbox{G-CURE}(\C,\bbeta,\bPhi; \Y, \X, \widetilde{\bO}^{(k)})$, $k=1,\ldots, r$, can be solved in parallel. 
{ GOFAR(P) obtains the final estimate of  ${\bbeta}$ and ${\bPhi}$ from the output of the $r$th (last) parallel procedure.}

It is clear that the quality of the initial estimator directly affects both the computational efficiency and the model accuracy of GOFAR(P). In practice, we recommend using either the mixed-outcome reduced-rank estimator proposed by \citet{chenandluo2017} when the model dimension is moderate or the lasso estimator when the model dimension is very high.

\begin{algorithm}[htp]
	\caption{Generalized Co-sparse Factor Regression via Parallel Extraction}
	\begin{algorithmic}\label{alg:geecure}
          \STATE Initialization:\\
          \begin{itemize}
          \item[(1)] Solve $\{\widetilde{\D},\widetilde{\U},\widetilde{\V},\widetilde{\bbeta},\widetilde{\bPhi} \}= \M{G-INIT}(\C,\bbeta,\bPhi; \Y, \X,\bO,r)$ and obtain $\widetilde{\C}_k$ (Section 1.3 of Supplementary Materials).
          \item[(2)] Compute offsets:  $\widetilde{\bO}^{(k)} = \bO + \X \sum_{i \neq k}\what{\C}_i$, for $k = 1,\ldots, r$.
         \end{itemize}   
		\FOR{$k \gets 1 \, \M{ to }\, r$}
	\STATE G-CURE with tuning (see Section \ref{sec2:computation}):\\
        $(\what{d}_k,\what{\u}_k,\what{\v}_k,\what{\bbeta},\what{\bPhi})  =  \mbox{G-CURE}(\C,\bbeta,\bPhi; \Y, \X,\widetilde{\bO}^{(k)}, \rho)$;\\ $\what{\C}_k  =  \what{d}_k\what{\u}_k\what{\v}_k\trans$.        \hspace{17em}\raisebox{.5\baselineskip}[0pt][0pt]{$\left.\rule{0pt}{2.2\baselineskip}\right\}\ \mbox{in parallel}$}
		\ENDFOR
                \RETURN{$\what{\C} = \sum_{k=1}^{r} \what{\C}_k$, $\what{\bbeta}$, $\what{\bPhi}$.}
	\end{algorithmic}
\end{algorithm}

\subsection{Generalized Co-Sparse Unit-Rank Estimation}\label{sec2:computation}

\subsubsection{Choice of Penalty Function}\label{subsec:penfun}

We denote the generic G-CURE problem as 
$$
\mbox{G-CURE}(\C,\bbeta,\bPhi; \Y, \X,\bO, \rho).
$$
First, we discuss the choice of the penalty function. In this work we use the elastic net penalty and its adaptive version \citep{zou2005,zou2009adaptive,mishra2017sequential}, i.e., for the $k$th step, 
\begin{align}
\rho(\C;\lambda) = \rho(\C;\W,\lambda,\alpha)& =\alpha\lambda\|\W\circ\C\|_1 + (1-\alpha)\lambda\|\C\|_F^2. \label{sec2:eq:enet} 
\end{align}
Here $\|\cdot\|_1$ denotes the $\ell_1$ norm, the operator ``$\circ$'' stands for the Hadamard product, $\W=[w_{ij}]_{p\times q}$ is a pre-specified weighting matrix, $\lambda$ is a tuning parameter controlling the overall amount of regularization, and $\alpha \in (0,1)$ controls the relative weights between the two penalty terms.  {
Several other penalties, such as the lasso $(\alpha = 1, \gamma = 0)$, the adaptive lasso $(\alpha = 1,  \gamma > 0)$, and the elastic net $(0 < \alpha < 1,\, \gamma = 0)$, are its special cases. }

In the $k$th step of GOFAR(S) or GOFAR(P), we let $\W_k = |\widetilde{\C}_k|^{-\gamma}$, where $\gamma =1$ and $\widetilde{\C}_k= \widetilde{d}_k\widetilde{\bu}_k\widetilde{\bv}_k\trans$ is an initial estimate of $\C_k$. As such, $w_{ijk} = w_k^{(d)} w_{ik}^{(u)} w_{jk}^{(v)}$, with
\begin{align}
w_k^{(d)} = |\widetilde{d}_k|^{-\gamma},
\bw_k^{(u)} = [w_{1k}^{(u)},..., w_{pk}^{(u)}]\trans= |\widetilde{\bu}_k|^{-\gamma},
\bw_k^{(v)} = [w_{1k}^{(v)},..., w_{qk}^{(v)}]\trans= |\widetilde{\bv}_k|^{-\gamma}. \label{sec2:weights}
\end{align}

Compared to lasso, a small amount of ridge penalty in the elastic net allows correlated predictors to be in or out of the model together, thereby improving the convexity of the problem and enhancing the stability of optimization \citep{zou2005,mishra2017sequential}; in our work, we fix $\alpha=0.95$ and write $\rho(\C;\W,\lambda,\alpha) = \rho(\C;\W,\lambda)$ for simplicity. Now we express $\mbox{G-CURE}(\C,\bbeta,\bPhi; \Y, \X,\bO,\rho)$ as
\begin{align}
(\hat{d}, \what{\u}, \what{\v}, \what{\bbeta},\what{\bPhi}) & \equiv \argmin_{\u,\d,\v, \bbeta, \bPhi}\,\,\,    \Big\{ F_{\lambda}(d,\u,\v,\bbeta,\bPhi) = \bbL(\bs\Theta,\bPhi) +  \rho(\C;\W, \lambda) \Big\},\label{eq:gcure:objT} \\ &\M{s.t.} \quad \C= d\u\v\trans, \, \u\trans\X\trans\X\u/n = \v\trans\v = 1, \bs\Theta = \bs\Theta(\C, \bbeta, \bO ). \notag 
\end{align}

\subsubsection{A Blockwise Coordinate Descent Algorithm}
To solve the problem in \eqref{eq:gcure:objT}, we propose an iterative algorithm  that cycles through a $\u$-step, a $\v$-step, a $\bbeta$-step and a $\bPhi$-step to update the unknown parameters in blocks of $(\u, d)$, $(\v,d)$, $\bbeta$ and $\bPhi$, respectively, until convergence. 
Below we describe each of these steps in detail.

\paragraph{$\u$-step}\label{subsec:ustep}
For fixed $\{\v, \bbeta, \bPhi\}$ with $\v\trans\v = 1$, we rewrite the objective function \eqref{eq:gcure:objT} in terms of the product variable $\check{\u} = d\u$  to avoid the quadratic constraints. For simplicity, we write $\bs\Theta(\C,\bbeta,\bO)$ as $\bs\Theta(\C)$. Motivated by \citet{she2012iterative} and \citet{chenandluo2017}, we construct a \emph{convex}  surrogate of the  the objective function \eqref{eq:gcure:objT} with respect to $\check{\u}$ as follows, 
{
\begin{align}
&G_{\lambda}(\a; \check{\u})=  \bbL(\bs\Theta(\a\v\trans),\bPhi)     + \Tr( \{\bB^'(\bs\Theta(\check{\u}\v\trans )) \}\trans \X(\a-\check{\u})\v\trans \bPhi^{-1} ) - \notag\\
&  \Tr(\J\trans [\bB(\bs\Theta(\a\v\trans))-\bB( \bs\Theta(\check{\u}\v\trans )  )]\bPhi^{-1})  + \frac{s_u}{2} \|\a-\check{\u}\|_2^2  +\rho(\a\v\trans;\W,\lambda) \notag \\ 
& = \frac{s_u}{2} \|\a-\check{\u} - \frac{\X\trans}{s_u}[\Y - \bB'( \bs\Theta(\check{\u}\v\trans )  )] \bPhi^{-1}\v  \|_2^2 + \rho(\a\v\trans;\W,\lambda) + \M{const}, 
 \label{eq:surrogateGmrrr} 
\end{align}
}where $s_u$ is a scaling factor for  the $\u$-step and ``const'' represents any { remaining} term that does not depend on the optimization variables; in this case, it is $\a \in \mathbb{R}^p$. 
It is easy to verify that $G_{\lambda}(\check{\u};\check{\u}) = F_{\lambda}(d,\u,\v,\bbeta,\bPhi)$. We show in the convergence analysis (see Section \ref{subsec:convanalysis}) that $F$ is majorized by $G_{\lambda}(\a;\check{\u})$ with appropriate scaling factor $s_u$. 
The problem of  minimizing $G_{\lambda}(\a;\check{\u})$ is separable in each entry of the vector $\a$. Hence, following \citet{zou2005},
the unique optimal solution 
is given by
\begin{align}
\what{\a} = \mathcal{\S}(\check{\u} + \X\trans[\Y & - \bB'( \bs\Theta(\check{\u}\v\trans/s_u )  )] \bPhi^{-1}\v; \notag \\ 
&\alpha \lambda\v\trans\w^{(v)} w^{(d)} \w^{(u)}/s_u )/ \{1+2\lambda(1-\alpha)\|\v\|_2^2/s_u\}, \label{eq:ustepupdate}
\end{align}
where  $\mathcal{\S}(\t;\tilde{\lambda}) = \M{sign}(\t)(|\t|-\tilde{\lambda})_+$ is the elementwise soft-thresholding operator on any $\t \in \mathbb{R}^p$. Now, using the equality constraint, i.e., $\|\X\u\|_2 = \sqrt{n}$, we can retrieve the individual estimates of $(d,\u)$ from $\what{\a}$.

\paragraph{$\v$-step}\label{subsec:vstep}
As in the $\u$-step, we rewrite the objective function \eqref{eq:gcure:objT} in terms of the product $\check{\v} = d\v$. A \emph{convex} surrogate that majorizes the objective function \eqref{eq:gcure:objT} with respect to $\check{\v}$ is constructed as 
\begin{align}
H_{\lambda}&(\b;\check{\v})=  \bbL(\bs\Theta(\u\b\trans),\bPhi)      + \Tr( \{ \bB^'(\bs\Theta(\u\check{\v}\trans )) \}\trans \X\u(\b-\check{\v})\trans \bPhi^{-1} ) - \notag\\
&  \Tr(\J\trans [\bB(\bs\Theta(\u\b\trans))-\bB( \bs\Theta(\u\check{\v}\trans )  )]\bPhi^{-1}) + \frac{s_v}{2} \|\b-\check{\v}\|_2^2   +\rho(\u\b\trans;\W,\lambda) \notag \\ 
 =& \frac{s_v}{2} \|\b-\check{\v} - \bPhi^{-1} [\Y - \bB'( \bs\Theta(\u\check{\v}\trans )  )]\trans \frac{\X}{s_v}\u  \|_2^2 + \rho(\u\b\trans;\W,\lambda) + \M{const},
 \label{eq:surrogateGmrrr} 
\end{align}
where $s_v$ is a scaling factor for  the $\v$-step and   $\b \in \mathbb{R}^q$ is the optimization variable. 
Following the $\u$-step, the unique optimal solution  minimizing $H_{\lambda}(\b;\check{\v})$ is given by 
\begin{align}
\what{\b} = \mathcal{\S}(\check{\v} + \bPhi^{-1} [\Y &- \bB'( \bs\Theta(\u\check{\v}\trans )  )]\trans \X\u/s_v; \notag \\ &\alpha \lambda\u\trans\w^{(u)} w^{(d)} \w^{(v)}/s_v )/ \{1+2\lambda_(1-\alpha)\|\u\|_2^2/s_v\}. \label{eq:vupdatestep}
\end{align}
Again, we retrieve the estimates of $(d,\v)$ from the equality constraint $\v\trans\v = 1$.

\paragraph{$\bbeta$-step}\label{subsec:betastep}
For fixed  $\C$ and $\bPhi$, denote $\bs\Theta(\bbeta) = \bs\Theta(\C,\bbeta,\bO)$. We construct a \emph{convex} surrogate that majorizes the objective function \eqref{eq:gcure:objT} with respect to $\bbeta$ as
{
\begin{align}
K(\balpha;\bbeta) = & \bbL(\bs\Theta(\balpha),\bPhi)   + \frac{s_{\beta}}{2} \|\balpha-\bbeta\|_2^2     +  \Tr( \{ \bB^'(\bs\Theta(\bbeta))\}\trans \Z(\balpha-\bbeta)\bPhi^{-1}) - \notag\\
&  \Tr(\J\trans [\bB(\bs\Theta(\balpha))-\bB( \bs\Theta( \bbeta )  )]\bPhi^{-1}) \notag \\ 
 =& \frac{s_{\beta}}{2} \|\balpha - \bbeta - \frac{\Z\trans}{s_{\beta}} \{\Y -  \bB^'(\bs\Theta(\bbeta))  \}\bPhi^{-1}\|_F^2   + \M{const},
 \label{eq:surrogateGmrrr} 
\end{align}
}where $s_{\beta}$ is a scaling factor  for  the $\bbeta$-step.

A globally optimal solution  minimizing $K(\balpha;\bbeta)$ is given by
\begin{align}
\what{\balpha} = \bbeta + \Z\trans \{\Y -  \bB^'(\bs\Theta(\bbeta))  \}\bPhi^{-1}/s_{\beta}.\label{eq:zupdaatestep}
\end{align}

\paragraph{$\bPhi$-step}\label{subsec:phistep}
For fixed $\C$ and $\bbeta$, we update $\bPhi$ by minimizing the negative log-likelihood function with respect to $\bPhi$, which can be obtained by a standard algorithm such as Newton-Raphson \citep{citer2019}.

The proposed G-CURE algorithm is summarized in Algorithm \ref{alg:spmrrr}.

\begin{algorithm}[!h]
  \caption{Generalized Co-Sparse Unit-Rank Estimation}
\begin{algorithmic}\label{alg:spmrrr}
\STATE Given: $\X$, $\Y$, $\Z$, $\W$, $\bO$, $\kappa_0$, $\lambda$, $\alpha$.
\STATE Initialize $\u^{(0)} = \widetilde{\u}$, $\v^{(0)} = \widetilde{\v}$, $d^{(0)} = \widetilde{d}$, $\bbeta^{(0)} = \widetilde{\bbeta}$, $\bPhi^{(0)} =  \widetilde{\bPhi}$. Set $t\gets0$.
\REPEAT
\STATE Set $s_u  =  \kappa_0\|\X\|^2/\varphi$, $s_{\beta} = \kappa_0\|\Z\|^2/\varphi$,  $s_v =   n\kappa_0/\varphi$ where $\varphi = \min(\bPhi^{(t)} )$. 
\vspace{0.1cm}
\STATE (1) $\u$-step:   Set $\check{\u} = d^{(t)}\u^{(t)}$ and $\v =\v^{(t)}$. Update $\check{\u}^{(t+1)}$ using  \eqref{eq:ustepupdate}. 
 Recover block variable  ($\tilde{d}^{(t+1)},\u^{(t+1)}$) using equality constraint in \eqref{eq:gcure:objT}.
\vspace{0.1cm}
\STATE (2) $\v$-step:  Set $\check{\v} = \tilde{d}^{(t+1)}\v^{(t)}$ and $\u =\u^{(t+1)}$. Update $\check{\v}^{(t+1)}$ using  \eqref{eq:vupdatestep}.
Recover block variable  ($d^{(t+1)},\v^{(t+1)}$) using equality constraint in \eqref{eq:gcure:objT}.
\vspace{0.1cm}
\STATE (3) $\bbeta$-step: Update $\bbeta^{(t+1)} $ using \eqref{eq:zupdaatestep}.
\vspace{0.1cm}
\STATE (4) $\bPhi$-step:
$\bPhi^{(t+1)} = \arg\max_{\bPhi} \bbL( \bs\Theta(\C^{(t+1)}, \bbeta^{(t+1)}),\bPhi)$. 
\vspace{0.1cm}
\STATE $t\gets t+1$.
\UNTIL {convergence, e.g., the relative $\ell_2$ change in parameters is less than $\epsilon = 10^{-6}$.}\\
\RETURN{$\widehat{\u}$, $\widehat{d}$, $\widehat{\v}$,  $\widehat{\bbeta}$, $\widehat{\bPhi}$.}
\end{algorithmic}
\end{algorithm}

\subsubsection{Convergence Analysis}\label{subsec:convanalysis}
In Algorithm \ref{alg:spmrrr}, we use several \emph{convex} surrogates of the objective function in order to deal with the general form of the loss function. { We show that  the procedure can ensure that the objective function is monotone descending with  the  scaling factors $s_u, \,\, s_v \M{ and } s_{\beta}$.} 

We mainly consider mixed outcomes of Gaussian, Bernoulli, and Poisson distributions as examples.
To conduct a formal convergence analysis, let us denote the parameter estimates  in the $t$th step as $\{ \u^{(t)}, d^{(t)},\v^{(t)},\bbeta^{(t)},\bPhi^{(t)}\}$. 
From Algorithm \ref{alg:spmrrr},  the $\u$-step produces ($ {\td}^{(t+1)}, \u^{(t+1)}$), the $\v$-step produces ($ d^{(t+1)}, \v^{(t+1)}$),  the $\bbeta$-step produces $\bbeta^{(t+1)}$, and the $\bPhi$-step produces $\bPhi^{(t+1)}$.

Now, denote $\check{\u}^{(t+1)} =\u^{(t+1)} {\td}^{(t+1)} $. For $\bs\xi_u^{(t+1)} \in \{a\check{\u}^{(t)}\v^{(t)}\trans +(1-a) \check{\u}^{(t+1)}\v^{(t)}\trans ; 0<a<1\}$ and $\bs\zeta(\bs \Theta_{.k}({\bs\xi}_u^{(t+1)}, \bbeta^{(t)} ),a_k(\phi_k^{(t)}) ) = \M{diag}[ \bB_{.k}^{''}(\bs\Theta_{.k}({\bs\xi}_u^{(t+1)}, \bbeta^{(t)} ))]/a_k(\phi_k^{(t)})$, we define
\begin{align*}
	\gamma_1^{(t)} = \sup_{a \in (0,1)} \| \X\trans \sum_{k=1}^q v_k^{(t)}^2 \bs\zeta( \bs\Theta_{.k}({\bs\xi}_u^{(t+1)}, \bbeta^{(t)} ),a_k(\phi_k^{(t)}) ) \X \|.  
	\end{align*}
Similarly, denote $\check{\v}^{(t+1)} =\v^{(t+1)} {d}^{(t+1)} $. Then, for $\bs\xi_v^{(t+1)} \in \{a{\u}^{(t)}\check{\v}^{(t)}\trans +(1-a) {\u}^{(t)}\check{\v}^{(t+1)}\trans ; 0<a<1\}$ and $\bs\zeta(\bs \Theta_{.k}({\bs\xi}_v^{(t+1)}, \bbeta^{(t)} ),a_k(\phi_k^{(t)}) ) = \M{diag}[ \bB_{.k}^{''}(\bs\Theta_{.k}({\bs\xi}_v^{(t+1)}, \bbeta^{(t)} ))]/a_k(\phi_k^{(t)})$, we define
\begin{align*}
	\gamma_2^{(t)} =   \max_{1\leq k \leq q} \sup_{a \in (0,1)}   \|\u^{(t)}\trans \X\trans   \bs\zeta( \bs\Theta_{.k}({\bs\xi}_v^{(t+1)}, \bbeta^{(t)} ),a_k(\phi_k^{(t)}) ) \X\u^{(t)}\|   
	\end{align*}
Finally, for $\bs\xi_{\beta}^{(t+1)} \in \{a\bbeta^{(t)} +(1-a) \bbeta^{(t+1)} ; 0<a<1\}$ and $\C^{(t+1)} = d^{(t+1)}\u^{(t+1)}\v^{(t+1)}\trans$, we define
	\begin{align*}
	\gamma_3^{(t)} =   \max_{1\leq k \leq q} \sup_{a \in (0,1)}\| \Z\trans \bs\zeta(\bs \Theta_{.k}({\C^{(t+1)}, \bs\xi}_{\beta}^{(t+1)}  ),a_k(\phi_k^{(t)}) )  \Z \|.
	\end{align*}

\begin{theorem}\label{th:convergence}
The sequence $\{d^{(t)},\u^{(t)},\v^{(t)},\bbeta^{(t)},\bPhi^{(t)}\}_{t \in \mathbb{N}}$ produced by Algorithm \ref{alg:spmrrr} satisfies 
	\begin{align*}
	F_{\lambda}(d^{(t)},& \u^{(t)},\v^{(t)},\bbeta^{(t)},\bPhi^{(t)}) \geq F_{\lambda}(d^{(t+1)}, \u^{(t+1)},\v^{(t+1)},\bbeta^{(t+1)},\bPhi^{(t+1)}),
	\end{align*}
	for the scaling factors   $ s_u \geq \gamma_1^{(t)}$, $s_v \geq \gamma_2^{(t)} $ and $s_{\beta} \geq  \gamma_3^{(t)} $.
\end{theorem}

 The proof of Theorem \ref{th:convergence} is relegated to Section 1.5 of Supplementary Materials. 
Further, we  follow \citet{she2012iterative} to obtain the  scaling factors $s_u, \, s_v \,\M{and}\, s_{\beta}$  that ensure that the objective function will be monotone decreasing  along the iterations. 
  The key is to find a good upper bound of $b_k''(x)$. It is known that for Gaussian responses, $b_k''(x) = 1$ and $a_k(\phi_k)=\sigma_k^2$, and for Bernoulli responses,  $b_k''(x) = e^x/(1+e^x)^2\leq 1/4$ and  $a_k(\phi_k)=1$. But for Poisson responses, $b_k''(x) = e^x$ is unbounded  and $a_k(\phi_k)=1$. Hence, in practice we  choose a large enough upper bound $\alpha_p$ of $b_k''(x)$ empirically (default $\alpha_p = 10$).

Now, based on the above discussion, define the upper bound $\kappa_0$ for $q$ outcomes such that $ b_k''(x) \leq \kappa_0 $ for all $k = 1,\ldots, q$. {   Then, at the $t$th step, we have $\gamma_1^{(t)} \leq \kappa_0\|\X\|^2/\min(a_k(\phi_k^{(t)}))$; $\gamma_2^{(t)} \leq \kappa_0\u^{(t)}\trans\X\trans\X\u^{(t)}\trans/\min(a_k(\phi_k^{(t)})) = n\kappa_0/\min(a_k(\phi_k^{(t)})) $; and $\gamma_3^{(t)} \leq \kappa_0\|\Z\|^2/\min(a_k(\phi_k^{(t)}))$. 
Hence, we set the scaling factors $s_u   = \kappa_0\|\X\|^2/\varphi$ for the $\u$-step, $s_v =  n\kappa_0/\varphi$ for the $\v$-step and $s_{\beta} = \kappa_0\|\Z\|^2/\varphi$ for the $\bbeta$-step where $\varphi = \min(a_k(\phi_k^{(t)}))$.}

Algorithm 3 is a block coordinate descent optimization procedure for minimizing the nonsmooth and nonconvex objective functions $F_{\lambda}(d,\u,\v,\bbeta,\bPhi)$.
This type of problem has been studied in, e.g., \citet{gorski2007}, \citet{razaviyayn2013unified} and \citet{mishra2017sequential}. In each of the sub-problems, the algorithm minimizes a \emph{convex} surrogate that majorizes the objective function, which results in a unique and bounded solution when the elastic net penalty is used  \citep{mishra2017sequential}. Thus, using Theorem/Corollary 2(a) of \citet{razaviyayn2013unified}, we can conclude that any limit point of the sequence of solutions generated by the algorithm is a coordinate-wise minimum of the objective function.
Algorithm \ref{alg:spmrrr} always converges in our extensive numerical studies. 
Both of the estimation procedures for GOFAR are implemented, tested, validated, and made  publicly available in a user-friendly R package,  \textsf{gofar}.

\subsubsection{Tuning and A Toy Example}\label{sec2:TUNING}

Using Algorithm \ref{alg:spmrrr}, we minimize $F_{\lambda}(d,\u,\v,\bbeta,\bPhi)$ over a range of $\lambda$ values while fixing $\alpha = 0.95$ and $\gamma = 1$.  The range of $\lambda$ (equispaced on the log-scale), i.e., $\lambda_{max}$ to $\lambda_{min}$, is chosen in order to produce a spectrum of possible sparsity patterns in $\u$ and $\v$. Specifically, $\lambda_{\max}$ is the smallest $\lambda$ at which the singular value estimate is zero. In practice, we choose $\lambda_{\max}=\|\X\trans(\Y-\bmu(\0)) \|_{\infty}$, and set $\lambda_{\min}$ as the fraction of $\lambda_{\max}$, i.e., $\lambda_{\min}=\lambda_{\max}\times 10^{-6}$,  at which the estimated singular vectors have larger support, i.e.,  nonzero entries,   than expected. 
The optimal $\lambda$ can then be selected by $K$-fold cross-validation \citep{stone1974cross}.

Figure \ref{fig:gsecure-pathv1} shows the solution paths in  simulation setup I with Gaussian-Binary responses; see Table \ref{tab:simsetup} for details.
The models on the solution paths are compared by the cross-validated negative log-likelihood. As with the implementation of \textsf{glmnet}, we suggest using the one-standard-deviation rule to select the final solution.

\begin{figure}[htp]
	\centering
		\subfloat[ $\C$ ]{{\includegraphics[width=0.33\textwidth]{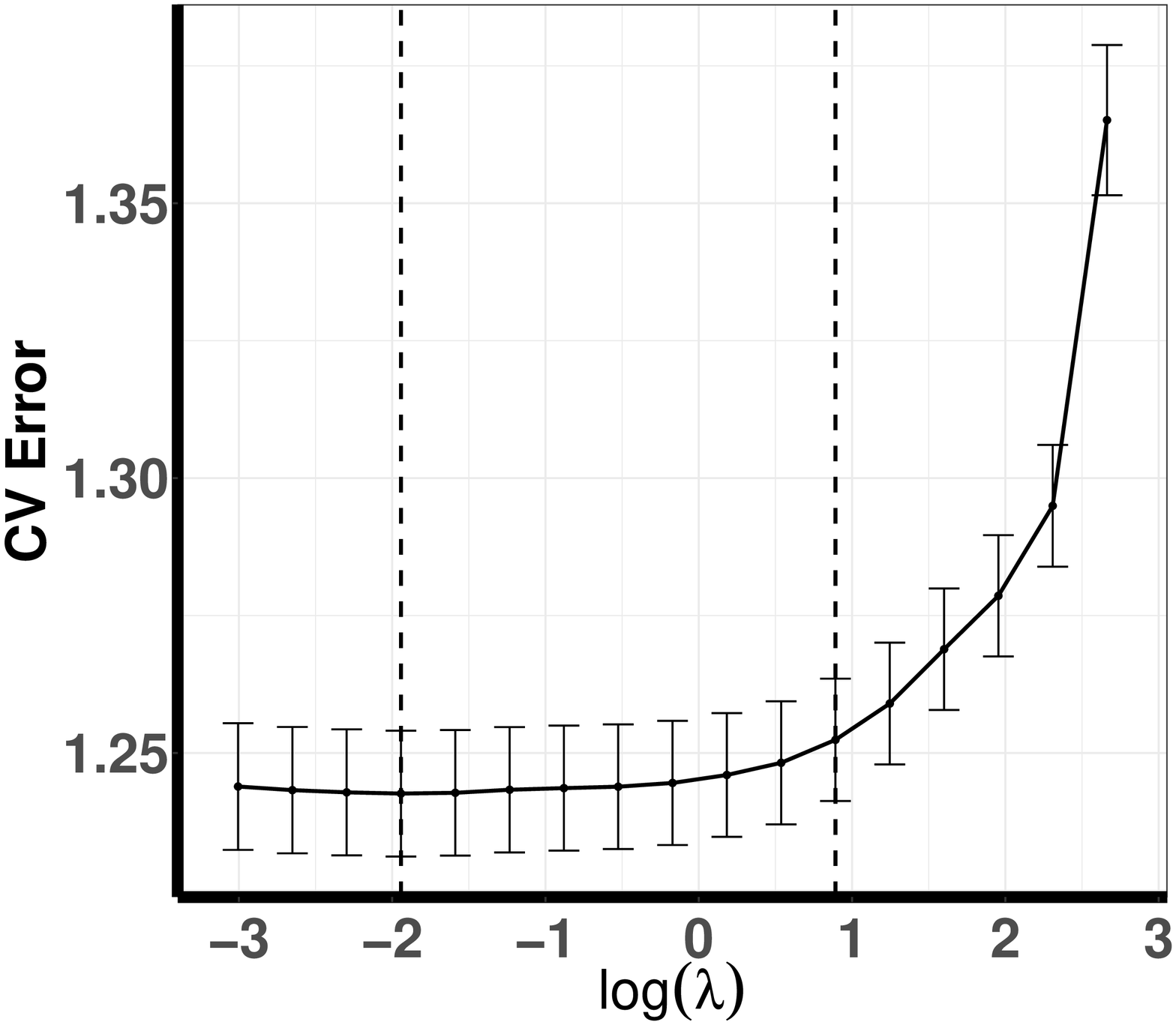} }}
		\subfloat[ $\C_1$]{{\includegraphics[width=0.33\textwidth]{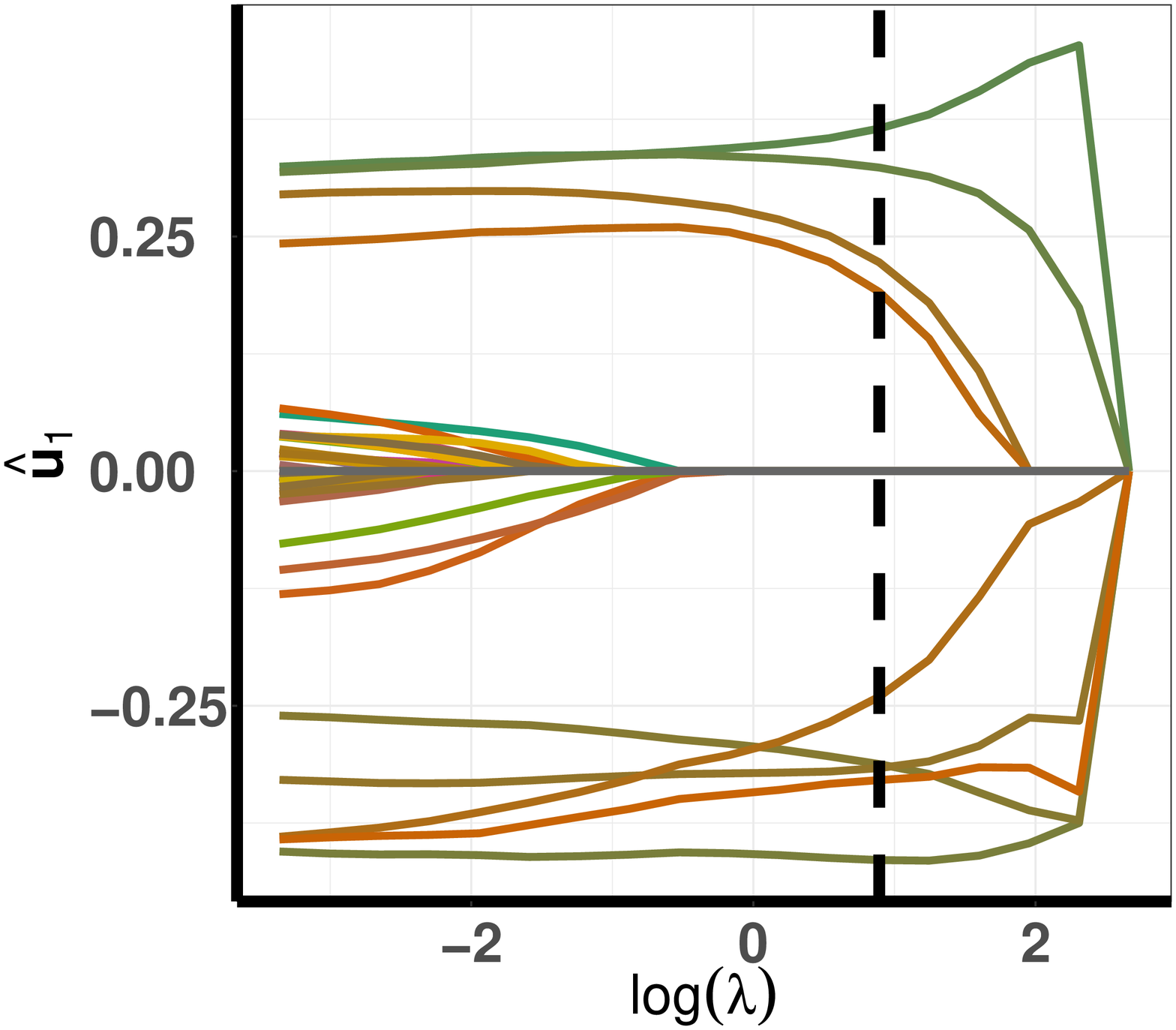} }}
			\subfloat[ $\C_2$]{{\includegraphics[width=0.33\textwidth]{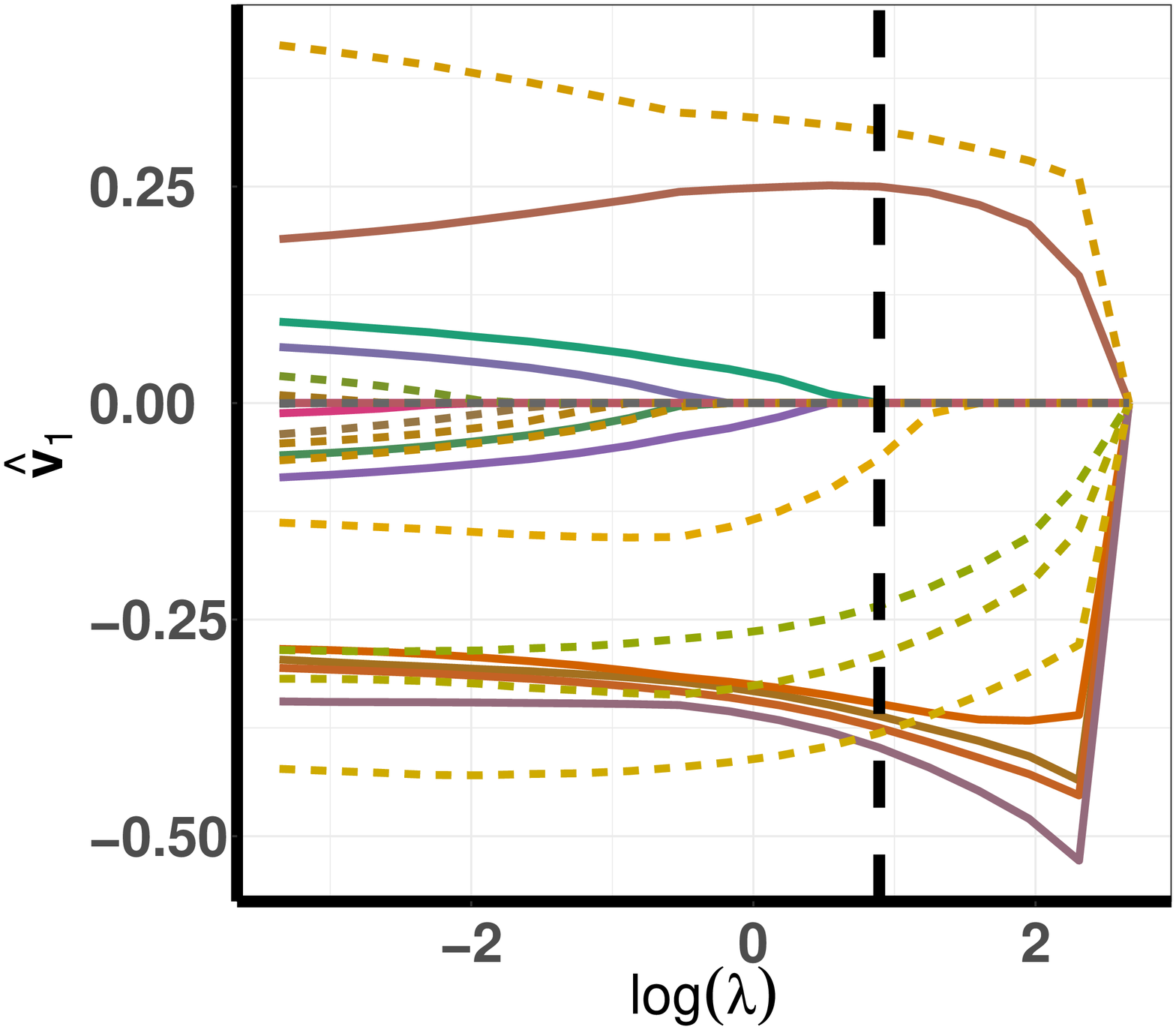} }}
		\caption{G-CURE: (a) cross-validation plot for 
		selecting the tuning parameter $\lambda$; (b)-(c) solution paths of $d\u$ and $d\v$, respectively, in case of simulation setup I with Gaussian-Binary responses; see Table \ref{tab:simsetup} for details.  The dashed and continuous lines in (c) differentiate between the two types of responses.}
	\label{fig:gsecure-pathv1}
\end{figure}

\section{Theoretical properties}\label{sec2:theory}
{ In order to focus on the large sample properties of the estimate of the unit-rank components of $\C$, we assume that the dispersion parameters $\Phi$ are known. Now, without loss of generality, we set $\Phi=\I$ and $\bO = \0$.}
Using the natural parameter $\bs\Theta^*$ formulated in equation \eqref{eq:ThetaDef} and the notations defined in equation \eqref{eq:defbtheta}, we  represent the multivariate model  \eqref{eq:glmmodel} for mixed outcomes as 
\begin{align}
\Y = \bB^'(\bs\Theta^*) + \E, \label{main:eq11}
\end{align}
where
\begin{itemize}
\item[\textbf{A1}.] the entries of the  error $\E = [e_{ik}]$  are independent $(\sigma^2 , b)\M{-sub-exponential}$ random variables with expectation $\mathbb{E}(e_{ij}) = 0$. 
\end{itemize}
In  large sample theory, we let $n$ tend to infinity with ($p,q$) fixed.
To ensure identifiability of the parameters, we make the following  assumptions on the covariates ($\X ,\Z$) and the true coefficient matrix $\C^*$. 
\begin{itemize}
	\item[\textbf{A2}.] $(1/n)\X\trans\X \xrightarrow{a.s} \bGamma_1$, $(1/n)\Z\trans\Z \xrightarrow{a.s} \bGamma_2$  and $(1/n)\X\trans\Z \xrightarrow{a.s} \0$ as $n \to \infty$, where $\bGamma_1$ and $\bGamma_2$ are fixed, positive definite matrices.
\end{itemize}
\begin{itemize}
	\item[\textbf{A3}.] $d_1^* > \ldots > d_{r^*}^*>0$.
\end{itemize}

To conveniently present our analysis, we allow each of the singular values $d_k^*$ to be absorbed into the pair $(\u_k^*,\v_k^*)$ of the decomposition \eqref{sec2:intro:cdef} \citep{chen2012jrssb}. Specifically, let $\ell_k$ denote the index of any nonzero entry  $\v_k^*$. Then, a uniquely identifiable reparameterization $\C_k^*$ is given by 
$$
\C_k^* = \u_k^*{\v_k^*}\trans, \qquad \mbox{s.t. } \qquad v_{\ell_k k}^* = 1.
$$
This results in $({\u_k^*}\trans\bGamma\u_k^*)({\v_k^*}\trans\v_k^*) = d_k^*$. 
Consequently, 
\begin{align}
\C^* =\U^*{\V^*}\trans, \quad \mbox{ s.t. } &{\U^*}\trans\bGamma\U^* \mbox{ and } {\V^*}\trans\V^* \mbox{ are both diagonal matrices},\notag\\
&  v_{\ell_k k}^* = 1, k=1,\ldots,r^*.\label{sec2:eq:cstar}
\end{align}

In terms of the new parameterization, the objective function of the G-CURE optimization problem \eqref{eq:gcure:objT} is given by 
\begin{align}
F_k^{(n)}(\u,\v ,\bbeta) &= \bbL(\C,\bbeta;\bO_k) + \rho(\C;\W_k, \lambda_k^{(n)}),  \label{objmsureKelpen}
\end{align}
where $\u \in \mathbb{R}^p$, $\v \in \mathbb{R}^q$ with $v_{\ell_k}=1$, $\C = \u\v\trans$, and  the offset matrix $\bO_k$ depends on the choice of the estimation procedure, i.e., GOFAR(S) or GOFAR(P), and mainly follows from  equations \eqref{eq:gcure:ofset} and \eqref{eq:offset-gofarp}. Here $\W_k = [ w_{ijk}]_{p \times q} = [w_{ik}w_{jk}]_{p \times q}$,  where $w_{ik} = |\widetilde{u}_{ik}|^{-\gamma}$ and $w_{jk} = |\widetilde{v}_{jk}|^{-\gamma}$ for some $\gamma >0$. The regularization parameter $\lambda_k^{(n)}$ is  a function of the sample size, but $0<\alpha\leq 1$ is considered as a fixed constant. 
In our model formulation,  $b_k''(\theta_{ik})$ corresponds to the variance of the estimate of the ${ik}$th outcome for $\theta_{ik}$.  
 Motivated by \citet{chenandluo2017}, we assume that 
\begin{itemize}
	\item[\textbf{A4}.] $b_k(\cdot)$ is a  continuously differentiable, real-valued and strictly convex function, and the entries of the natural parameter $\bs\Theta$ defined in \eqref{eq:ThetaDef}  satisfy
	$$ \min_{\substack{1\leq i \leq n \\ 1\leq k \leq q}} \quad \inf_{\{\bbeta , \C\}} \,\, |b_k''(\theta_{ik})| \geq \gammaL, $$
\end{itemize}
for some constant $\gammaL > 0$.

Moreover, GOFAR(P) requires an initial estimate  of the unit-rank components  of the rank-$r$ coefficient matrix $\C$, given by  $\widetilde{\C}_i$ for $i = 1,\ldots,r$. We require the initial estimators to be  $\sqrt{n}$-consistent, i.e., 
\begin{itemize}
	\item[\textbf{A5}.] 
	$ \| \widetilde{\C}_i - \C_i^* \| =O_p(n^{-1/2}) $ for $i = 1,\ldots,r$.
\end{itemize}
This can be achieved by the unpenalized GLM estimators or the reduced-rank estimator \citep{reinsel1998,chenandluo2017}, although these estimators do not have the desired sparse SVD structure.

\begin{theorem}\label{Sec2:TH:localminima}
Assume \textbf{A1}--\textbf{A5} hold and  $\lambda_k^{(n)} /\sqrt{n} \to \lambda_k \geq 0$ as $n \to \infty$. 
Then the estimator $(\what{\u}_k,\what{\v}_k, \what{\bbeta})$, from either the sequential or the parallel estimation, is $\sqrt{n}$-consistent, i.e., 
	\begin{enumerate}[i.]
		\item $\|\what{\u}_k - \u_k^* \| = O_p(n^{-1/2})$, $\|\what{\v}_{k} - \v_{k}^* \| = O_p(n^{-1/2})$, and $\|\what{\bbeta} - \bbeta^* \| = O_p(n^{-1/2})$  for $k=1,\ldots,r^*$.
		\item $|\what{d}_k|=O_p(n^{-1/2})$ where $\what{d}_k = (1/n)(\what{\u}_k\trans\X\trans\X\what{\u}_k)(\what{\v}_k\trans\what{\v}_k)$, for $k=r^*+1,\ldots,r$.
	\end{enumerate}
\end{theorem}
Here, we have mainly followed the setup of \citet{mishra2017sequential} to prove the required results,  the details of which are relegated to Supplementary Materials,  Section  1.6. Similarly, by following \citet{mishra2017sequential}, we can establish the selection consistency of  GOFAR(S) and GOFAR(P) under assumptions \textbf{A1 - A5}.


\section{Simulation}\label{sec2:simulation}
\subsection{Setup}\label{subsec2:simsetup}
We compare the estimation performance, prediction accuracy and sparsity recovery of GOFAR(S) and  GOFAR(P) to those of the following modeling strategies: (a) uGLM: fit each response by the univariate sparse GLM implemented in the R package \textsf{glmnet}  \citep{friedman2010regularization}; and (b) mRRR: fit by mixed-outcome reduced-rank regression \citep{chenandluo2017}. In addition, to show the merit of jointly learning from mixed outcomes, we also use GOFAR(S) to fit each type of responses  separately;  the resulting method is labeled GOFAR(S,S).

We have summarized all the simulation settings in Table \ref{tab:simsetup}.  The  setup covers scenarios with the same type of outcomes and with mixed types of outcomes. In the first scenario, the outcomes are  either Gaussian (G), Bernoulli (B) or Poisson (P), whereas in the second scenario, the outcomes consist of an equal number of (a) Gaussian and Bernoulli (G-B) or (b) Gaussian and Poisson (G-P) outcomes. Moreover, setup I and setup II refer to the low-dimensional and high-dimensional simulation examples, respectively.

\begin{table}[htp]\label{tab:simsetup}
\centering
\caption{Simulation: model dimensions of all the simulation settings, including the sample size $n$, the number of predictors $p$, and the numbers \{$q_1$, $q_2$, $q_3$\} of Gaussian (G), Bernoulli (B) and Poisson (P) outcomes, respectively.} \label{tab:simsetup}
\scalebox{0.9}{
\begin{tabular}{cccccccc}
  \hline
\multicolumn{1}{c}{} & \multicolumn{2}{c}{} & \multicolumn{3}{c}{Single-Type Scenario} & \multicolumn{2}{c}{Mixed-Type Scenario}\\
\hline
Setup & $n$ & $p$ & G  & B & P & G-B &  G-P\\
\hline
I &  200 & 100 & (30,0,0)  & (0,30,0) & (0,0,30)  & (15,15,0)  & (15,0,15)\\  
II &  200 & 300 & (30,0,0) &  (0,30,0) & (0,0,30)  & (15,15,0)  & (15,0,15)\\
\hline
\end{tabular}}
\end{table}

We set the true rank as $r^* =3$. Denote the true coefficient matrix as $\C^* = \U^*\D^*\V^*\trans$, with $\U^* = [\u_1^*,\u_2^*,\u_3^*]$, $\V^* = [\v_1^*,\v_2^*,\v_3^*]$ and $\D^* = s \times \M{diag}[d_1^*,d_2^*,d_{3}^*]$. We set $d_1^*=6$, $d_2^*=5$, $d_3^*=4$ and $s=1$, except that when Poisson outcomes are present we set $s = 0.4$. 
The particular choice of the default value of $\alpha_p = 10$ for Poisson outcomes ensures a  monotone descending objective function for the G-CURE optimization problem \eqref{eq:gcure:objT}. Let $\M{unif}(\mA,b)$ denote a vector of length $b$ whose entries are 
uniformly distributed on the set $\mA$, and $\M{rep}(a,b)$ denote the vector of length $b$ with all entries  equal to $a$. For the single-type response scenario, we generate $\u_k^*$ as $\B{u}_k^*=\check{\B{u}}_k/\|\check{\B{u}}_k\|$, where $\check{\B{u}}_1=[\M{unif}(\mA_u,8), \M{rep} (0, p-8)]\trans$, $\check{\B{u}}_2=[\M{rep}(0,5), \M{unif}(\mA_u,9),\M{rep}(0,p-14)]\trans$, and $\check{\B{u}}_3=[\M{rep}(0,11),\M{unif}(\mA_u,9),\M{rep} (0,p-20)]\trans$; and we generate $\v_k^*$ as $\B{v}_k^*=\check{\B{v}}_k/\|\check{\B{v}}_k\|$, where $\check{\B{v}}_1=[\M{unif}(\mA_v,5),\M{rep}(0,q-5)]\trans$, $\check{\B{v}}_2=[\M{rep}(0,5), \M{unif}(\mA_v,5),$ $\M{rep}(0,q-10)]\trans$, and $\check{\B{v}}_3=[\M{rep}(0,10)$, $\M{unif}(\mA_v,5),\M{rep}(0,q-15)]\trans$. Here we set $\mA_u=\pm 1$ and $\mA_v=[-1,-0.3]\cup[0.3,1]$. For the mixed-type scenario, while the $\u_k^*$s are generated in the same way, we set  the $\v_k^*$s to make sure there is sufficient sharing of information among the different types of responses. Specifically, we generate $\v_k^*$ as $\check{\B{v}}_k = [\widebar{\v}_{k},\widebar{\v}_{k}]\trans$ for $k=1,2,3$, where $\widebar{\v}_1=[\M{unif}(\mA_u,5),\M{rep}(0,q/2-5)]$, $\widebar{\v}_2=[\M{rep}(0,3),\widebar{v}_{14},-\widebar{v}_{15}, \M{unif}(\mA_u,3), \M{rep}(0,q/2-8)]$, and $\widebar{\v}_3=[\widebar{v}_{11},-\widebar{v}_{12},\M{rep}(0,4),\widebar{v}_{27},-\widebar{v}_{28},\M{unif}(\mA_u,2),\M{rep}(0,q-10)]$. In all the settings, we set $\Z = \1_n$ with $\bbeta^* = [\M{rep}(0.5,q)]\trans$, to include an intercept term.

The predictor matrix $\X \in\mathbb{R}^{n\times p}$ is generated from a multivariate normal distribution with some rotations to make sure that the latent factors $\X\U^*/\sqrt{n}$ are orthogonal according to the proposed GOFAR model; the details can be found in \citet{mishra2017sequential}. The dispersion parameter $a_k(\phi_k^*) = \sigma^2$ for the Gaussian outcomes is set to make the signal-to-noise ratio (SNR) equal to 0.5. (For the Binary and Poisson outcomes, $a_k(\phi_k^*) = 1$). Finally, $\Y$ is generated according to model \eqref{eq:glmmodel} with $\bs\Theta^* = \Z\bbeta^* + \X\C^*$.  We also consider the incomplete data setup by randomly deleting 20\% of the entries in $\Y$ ($\mbox{M}\% = 20$). The experiment under each setting is replicated 100 times.

The model estimation performance is measured by $\M{Er}({\C}) = \| \what{\C} - \C^*\|_F/(pq)$ and $\M{Er}(\bs\Theta) = \| \what{\bs\Theta} - \bs\Theta^*\|_F /(nq)$. 
The sparsity recovery is evaluated by the false positive rate (FPR) and the false negative rate (FNR), calculated by comparing the support of $(\hat{\u}_k,\hat{\v}_k)$ to that of $(\u_k^*,\v_k^*)$ for $k=1,\ldots,r^*$. For rank recovery, we report the mean of the rank estimates and the relative percentage of signal in the $(r^*+1)$th component and beyond, i.e., $\M{R}\% = 100(\sum_{i = r^* +1}^{\what{r}}\hat{d}_i^2) /  (\sum_{i = 1}^{{\what{r}}}\hat{d}_i^2) $; as such, $\M{R}\% = 0$ if the rank is not over-estimated. Finally, we depict the  computational complexity in terms of  mean execution time.

\subsection{Simulation Results}\label{sec::sim-res}

Tables \ref{table:respTypeBernoulliII}--\ref{table:respTypeGaussian-PoissonII} report the results for the high-dimensional models in Setup II (Table \ref{tab:simsetup}). Figures \ref{fig:similartype:analysis}--\ref{fig:mixedtype:analysis} show the boxplots of the estimation errors for Setups I and II. The detailed results under Setup I are relegated to  Supplementary Materials, as the results under the two setups convey similar messages. 

\FloatBarrier 
\begin{table}[!h]
\centering
\parbox{1.00\textwidth}{\caption{Simulation: model evaluation based on 100 replications using various performance measures (standard deviations are shown in parentheses) in Setup II with Bernoulli responses. \label{table:respTypeBernoulliII}}} 
\begingroup\scriptsize
\scalebox{0.75}{
\begin{tabular}{llllllll}
  \hline
 & Er($\C$) & Er($\bs\Theta$) & FPR & FNR & R\% & r & time (s) \\  
   \hline 
& \multicolumn{7}{c}{ M\% = 0} \\
 \hline
GOFAR(S) & 22.59 (5.15) & 41.29 (6.59) & 2.40 (0.98) & 4.00 (2.95) & 0.00 (0.00) & 3.00 (0.00) & 247.21 (13.82) \\ 
  GOFAR(P) & 26.08 (6.94) & 55.10 (14.65) & 6.26 (2.70) & 3.30 (2.98) & 0.00 (0.00) & 3.00 (0.00) & 49.34 (6.96) \\ 
  mRRR & 149.30 (12.27) & 272.49 (27.92) & 100.00 (0.00) & 0.00 (0.00) & 0.00 (0.00) & 3.00 (0.00) & 51.17 (0.80) \\ 
  uGLM & 58.11 (2.92) & 120.95 (6.69) & 72.56 (5.68) & 1.40 (1.47) & 18.17 (2.74) & 25.14 (1.51) & 8.05 (0.19) \\ 
   \hline 
& \multicolumn{7}{c}{ M\% = 20} \\
GOFAR(S) & 28.69 (4.96) & 54.95 (7.36) & 2.74 (0.95) & 7.66 (4.19) & 0.00 (0.00) & 3.00 (0.00) & 274.09 (15.10) \\ 
  GOFAR(P) & 38.03 (9.65) & 88.77 (25.48) & 8.81 (3.37) & 5.63 (4.15) & 0.00 (0.00) & 3.00 (0.00) & 54.54 (7.14) \\ 
  mRRR & 150.53 (24.42) & 307.31 (48.92) & 81.65 (20.33) & 17.96 (19.89) & 0.00 (0.00) & 2.45 (0.61) & 50.96 (0.93) \\ 
  uGLM & 63.93 (2.90) & 140.12 (8.66) & 67.57 (6.81) & 2.94 (2.71) & 28.70 (12.66) & 24.72 (1.62) & 5.90 (0.19) \\ 
   \hline
\end{tabular}
}
\endgroup
\end{table}

\FloatBarrier 
\begin{table}[!h]
\centering
\parbox{1.00\textwidth}{\caption{Simulation: model evaluation based on 100 replications using various performance measures (standard deviations are shown in parentheses) in Setup II with Poisson responses. \label{table:respTypePoissonII}}} 
\begingroup\scriptsize
\scalebox{0.75}{
\begin{tabular}{llllllll}
  \hline
 & Er($\C$) & Er($\bs\Theta$) & FPR & FNR & R\% & r & time (s) \\  
   \hline 
& \multicolumn{7}{c}{ M\% = 0} \\
 \hline
GOFAR(S) & 2.22 (0.60) & 3.86 (0.73) & 0.53 (0.50) & 1.85 (2.39) & 0.00 (0.00) & 3.00 (0.00) & 815.40 (37.25) \\ 
  GOFAR(P) & 2.22 (0.59) & 3.97 (0.72) & 6.80 (3.46) & 0.91 (1.38) & 0.07 (0.16) & 3.59 (0.77) & 188.18 (7.01) \\ 
  mRRR & 12.10 (0.39) & 10.64 (0.59) & 100.00 (0.00) & 0.00 (0.00) & 11.74 (2.36) & 4.00 (0.00) & 54.26 (0.98) \\ 
  uGLM & 5.93 (0.69) & 10.28 (0.79) & 84.65 (4.44) & 0.00 (0.00) & 10.46 (1.66) & 25.57 (1.46) & 17.03 (0.71) \\ 
   \hline 
& \multicolumn{7}{c}{ M\% = 20} \\
GOFAR(S) & 2.74 (0.67) & 4.84 (0.96) & 0.67 (0.51) & 3.51 (2.93) & 0.00 (0.00) & 3.00 (0.00) & 846.06 (47.99) \\ 
  GOFAR(P) & 3.00 (0.77) & 5.18 (0.98) & 9.10 (4.22) & 1.21 (1.41) & 1.37 (1.73) & 3.69 (0.77) & 197.56 (6.14) \\ 
  mRRR & 13.04 (0.53) & 14.95 (2.47) & 100.00 (0.00) & 0.00 (0.00) & 8.63 (6.13) & 3.67 (0.61) & 54.85 (1.25) \\ 
  uGLM & 7.22 (0.72) & 13.04 (0.94) & 81.50 (4.86) & 1.18 (1.49) & 13.14 (2.07) & 25.34 (1.46) & 12.32 (0.44) \\ 
   \hline
\end{tabular}
}
\endgroup
\end{table}

Both GOFAR(S) and GOFAR(P) consistently outperform the other methods in terms of estimation accuracy, sparsity recovery, and rank identification at the expense of reasonably manageable execution time.  In particular, we observe that GOFAR methods maintain their superiority over the other competing methods for handling incomplete data; compared to the complete data counterpart, there is only a mild deterioration in the model estimator evaluation statistics. GOFAR(P) tends to have slightly better performance, which may be owing to the use of an offset that accounts for all the information of the non-targeted unit-rank components. So depending on the computational resources, one can use either of the approaches.

\FloatBarrier 
\begin{table}[!h]
\centering
\parbox{1.00\textwidth}{\caption{Simulation: model evaluation based on 100 replications using various performance measures (standard deviations are shown in parentheses) in Setup II with Gaussian-Bernoulli responses. \label{table:respTypeGaussian-BernoulliII}}} 
\begingroup\scriptsize
\scalebox{0.75}{
\begin{tabular}{llllllll}
  \hline
 & Er($\C$) & Er($\bs\Theta$) & FPR & FNR & R\% & r & time (s) \\  
   \hline 
& \multicolumn{7}{c}{ M\% = 0} \\
 \hline
GOFAR(S) & 20.49 (2.97) & 43.34 (5.23) & 0.31 (0.29) & 0.88 (1.21) & 0.00 (0.00) & 3.00 (0.00) & 129.46 (16.65) \\ 
  GOFAR(P) & 14.64 (4.59) & 29.60 (7.83) & 6.25 (3.06) & 1.54 (1.96) & 0.00 (0.00) & 3.00 (0.00) & 29.82 (4.13) \\ 
  mRRR & 76.17 (8.10) & 164.50 (31.49) & 33.37 (0.00) & 67.24 (0.00) & 0.00 (0.00) & 1.00 (0.00) & 53.74 (1.00) \\ 
  uGLM & 45.57 (2.54) & 86.09 (5.65) & 81.40 (4.66) & 0.00 (0.00) & 12.69 (1.52) & 23.94 (1.50) & 7.74 (0.16) \\ 
  GOFAR(S,S) & 34.14 (5.44) & 80.96 (14.23) & 24.15 (5.75) & 6.33 (4.83) & 7.27 (9.94) & 6.66 (1.37) & 211.37 (45.36) \\ 
   \hline 
& \multicolumn{7}{c}{ M\% = 20} \\
GOFAR(S) & 24.64 (4.11) & 51.56 (6.75) & 0.38 (0.27) & 2.25 (1.97) & 0.00 (0.00) & 3.00 (0.00) & 156.40 (19.51) \\ 
  GOFAR(P) & 20.50 (6.32) & 42.40 (11.35) & 11.67 (5.01) & 3.61 (3.65) & 0.67 (1.33) & 3.39 (0.65) & 34.80 (3.25) \\ 
  mRRR & 79.19 (6.96) & 171.80 (37.15) & 39.26 (12.75) & 60.80 (13.92) & 0.00 (0.00) & 1.18 (0.38) & 54.77 (1.01) \\ 
  uGLM & 51.89 (2.54) & 102.61 (7.21) & 79.26 (5.37) & 0.00 (0.00) & 15.85 (2.09) & 24.06 (1.64) & 5.64 (0.14) \\ 
  GOFAR(S,S) & 38.30 (4.61) & 90.68 (14.64) & 20.80 (5.80) & 7.73 (5.53) & 3.73 (1.83) & 6.00 (1.36) & 193.92 (46.18) \\ 
   \hline
\end{tabular}
}
\endgroup
\end{table}

\FloatBarrier 
\begin{table}[!h]
\centering
\parbox{1.00\textwidth}{\caption{Simulation: model evaluation based on 100 replications using various performance measures (standard deviations are shown in parentheses) in Setup II with Gaussian-Poisson responses. \label{table:respTypeGaussian-PoissonII}}} 
\begingroup\scriptsize
\scalebox{0.75}{
\begin{tabular}{llllllll}
  \hline
 & Er($\C$) & Er($\bs\Theta$) & FPR & FNR & R\% & r & time (s) \\  
   \hline 
& \multicolumn{7}{c}{ M\% = 0} \\
 \hline
GOFAR(S) & 2.26 (0.50) & 3.66 (0.64) & 0.32 (0.28) & 0.71 (1.06) & 0.00 (0.00) & 3.00 (0.00) & 686.87 (23.35) \\ 
  GOFAR(P) & 2.00 (0.54) & 3.08 (0.55) & 7.75 (4.83) & 0.00 (0.00) & 0.00 (0.00) & 3.00 (0.00) & 148.44 (5.12) \\ 
  mRRR & 13.88 (0.63) & 32.68 (2.52) & 33.37 (0.00) & 67.24 (0.00) & 0.00 (0.00) & 1.00 (0.00) & 57.01 (1.19) \\ 
  uGLM & 5.93 (0.57) & 9.31 (0.60) & 87.09 (3.12) & 0.00 (0.00) & 10.65 (1.33) & 24.03 (1.66) & 12.38 (0.38) \\ 
  GOFAR(S,S) & 3.69 (1.28) & 8.01 (3.60) & 14.16 (4.76) & 0.56 (1.05) & 38.72 (38.17) & 5.34 (0.89) & 493.27 (50.71) \\ 
   \hline 
& \multicolumn{7}{c}{ M\% = 20} \\
GOFAR(S) & 2.77 (0.58) & 4.58 (0.88) & 0.55 (0.46) & 1.39 (1.51) & 0.00 (0.00) & 3.00 (0.00) & 678.94 (24.96) \\ 
  GOFAR(P) & 3.02 (0.72) & 4.21 (0.78) & 15.02 (6.13) & 0.00 (0.00) & 0.00 (0.00) & 3.00 (0.00) & 147.07 (5.04) \\ 
  mRRR & 14.81 (0.53) & 35.89 (2.47) & 33.37 (0.00) & 67.24 (0.00) & 0.00 (0.00) & 1.00 (0.00) & 57.14 (1.16) \\ 
  uGLM & 7.08 (0.52) & 11.84 (0.77) & 83.87 (3.61) & 0.00 (0.00) & 13.51 (1.73) & 24.09 (1.68) & 8.90 (0.29) \\ 
  GOFAR(S,S) & 3.64 (1.02) & 6.29 (1.92) & 16.46 (4.72) & 0.59 (1.15) & 37.48 (37.06) & 5.98 (0.90) & 533.82 (47.11) \\ 
   \hline
\end{tabular}
}
\endgroup
\end{table}

\vspace{1cm}
\newpage
\begin{figure}[htp]
 \centering
  \includegraphics[width=1.0\textwidth]{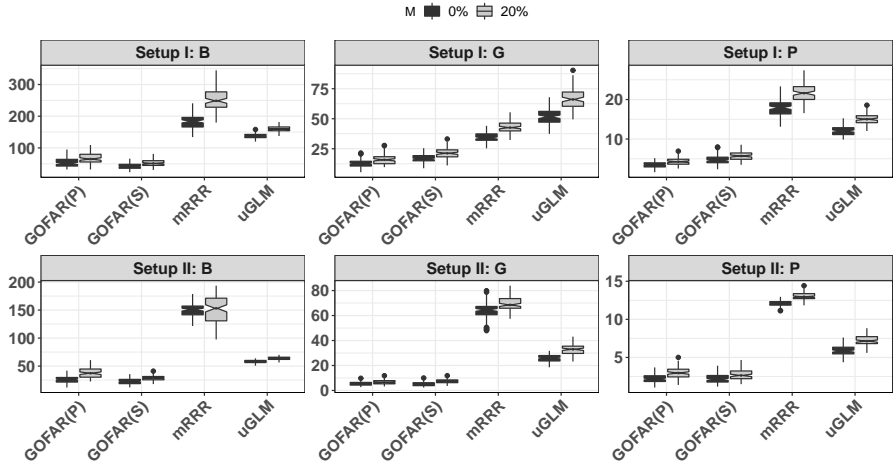}
      \caption{Simulation: notched boxplots of the estimation error $\M{Er}({\C})$ for the single-type scenario under Setups I and II based on 100 replications} 
     \label{fig:similartype:analysis}
\end{figure}

The superior performance of GOFAR  is due to its ability to model  the underlying association between multivariate responses and high-dimensional predictors through the low-rank and sparse coefficient matrix. On the other hand, the mRRR is only equipped to handle dependency through the low-rank structure. Because of this, the noise variables are all used in the estimated factors, thereby compromising the performance of the model; it may fail to identify important factors due to this limitation, which may cause rank underestimation, particularly in the mixed-type scenario.
The uGLM does not explore the shared information among the outcomes, while GOFAR(S,S) does not explore the shared information among the different types of responses; so the superior performance of GOFAR over these two models further showcases the merit of integrative multivariate learning.

\begin{figure}[htp]
 \centering
  \includegraphics[width=1.0\textwidth]{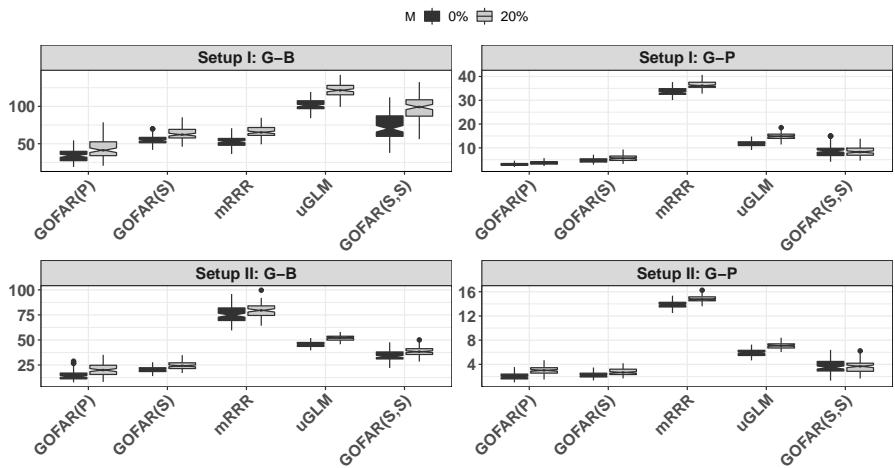}
      \caption{Simulation: notched boxplots of the estimation error $\M{Er}({\C})$ for the mixed-type scenario under  Setups I and II based on 100 replications.} 
     \label{fig:mixedtype:analysis}%
\end{figure}

\section{Application}\label{sec:gsecure-application}
\subsection{Modeling of mixed outcomes from LSOA}
The Longitudinal Study of Aging (LSOA) \citep{stanziano2010review}, a joint project of the National Center for Health Statistics and  the National Institute on Aging,  was designed to collect data measuring medical conditions, functional status, experiences and other socioeconomic dimensions of health in an aging population  (70 years of age and over). The study collected data from a large cohort of senior people in the period of  $1997-1998$. They were  studied again between $1999-2000$. 
Our goal is to understand the association between health-related events in the future (denoted as outcome $\Y$) and health status in the past (denoted as predictor $\X$) using data from $n = 3988$ subjects from this study.

The multivariate responses in $\Y$ include: a) $q_1 = 3$ continuous outcomes related to overall health status, memory status and depression status; and b) $q_2 =41$ binary/Bernoulli  outcomes related to physical conditions, medical issues, memory status, vision and hearing status, and social behavior. 
Our analysis considers a total of $p=294$ predictors, constructed from the variables related to demography, family structure, daily personal care, medical history, social activity, health opinion, behavior, nutrition, health insurance, income, and assets, a majority of which are binary measurements. For simplicity, we impute missing entries in the predictors with the sample mean. 
GOFAR(S)/GOFAR(P) can efficiently handle missing entries in the multivariate response, so such imputations are not required for the  20.2\% of   entries in $\Y$ that are missing.
Now, to determine the  association between $\X$ and $\Y$, we model the mixed outcomes jointly and apply  GOFAR(S)/GOFAR(P) to obtain a low-rank and sparse estimate of the coefficient matrix. The model specifies   gender and age as control variables $\Z$ (not penalized in the model). The parameter estimates  then relate a subset of future health outcomes to a subset of past health conditions via latent factors (constructed from the subset of predictors).

On the LSOA data,  {GOFAR(S)}/{GOFAR(P)}  demonstrates comparable prediction  performance with the advantage of producing the most parsimonious model when compared with the non-sparse method mRRR and the marginal approach uGLM. Table  \ref{tab:App:LSOA:PECI} summarizes the results from 100 replications with 75\%  of data selected using random sampling without replacement for training and the remaining 25\% for testing. On the test data, the metric Er(G) computes the mean square error for Gaussian outcomes and the metric  Er(B) computes the area under curve for binary outcomes. With the lesser number of latent factors (from r) and sufficiently sparse left and right singular vectors, GOFAR(S) produces the most parsimonious model, thus facilitating better interpretation.

\begin{table}[htp]
\centering
\caption{Application -- LSOA: Model evaluation (standard deviations are shown in parentheses) based on prediction error of Gaussian and binary outcomes, rank estimation $r$ and support recovery \{\textsf{supp}($\U$) and \textsf{supp}($\V$)\}.}\label{tab:App:LSOA:PECI}
\begingroup\scriptsize
\scalebox{1.10}{
\begin{tabular}{llllll}
  \hline
Method & Er(G) & Er(B) & r & \textsf{supp}(U)\{\%\} & \textsf{supp}(V)\{\%\} \\ 
  \hline
GOFAR(S) & 0.69(0.06) & 0.76(0.10) & 4.45(0.65) & 20(3) & 43(5) \\ 
  GOFAR(P) & 0.72(0.06) & 0.76(0.10) & 4.46(0.50) & 25(3) & 51(5) \\ 
  mRRR & 0.70(0.06) & 0.74(0.10) & 13.15(1.75) & 100(0) & 100(0) \\ 
  uGLM & 0.68(0.06) & 0.78(0.08) & 41.51(0.61) & 72(2) & 99(0) \\ 
   \hline
\end{tabular}
}
\endgroup
\end{table}

Owing to the superior performance of GOFAR(S) in terms of producing the most interpretable model, we apply the procedure to the full data and obtain the parameter estimates. The 
GOFAR(S) approach identifies $r=5$ subsets of outcome variables (inferred from the sparse $\V$) that are associated with the predictor $\X$ via an equivalent number of latent factors (constructed from a subset of predictors using the sparse $\U$). 
Figure  \ref{fig:mixedtype:lsoa-coef}  displays the sparse estimate of the coefficient matrix $\C$ and its $r$ unit-rank components. 
Support of the estimate of the singular vectors  is given by $\M{supp}(\U) = \{ 16\% , 30\%, 34\% , 54\%, 6\%\}$ and $\M{supp}(\V) = \{ 86\% , 72\%, 34\% , 14\%, 9\%\}$. 
The block structure of the unit-rank components facilitates a similar interpretation, as expected from biclustering. First, latent factors constructed from a subset of predictors, mainly in the category of daily activity and prior medical conditions, determine all outcomes except cognitive ability.
The latent factor clearly distinguishes social involvement outcomes from others. Apart from identifying the subset of predictors in each category, the approach finds a subgroup of the prior medical conditions affecting the outcome in the opposite way. The second latent factor helps us to identify a subgroup of the fundamental daily activity outcomes. 
One of the subgroups is similar to the group of  outcomes related to social involvement and medical conditions. The third and fourth latent factors clearly distinguish outcomes related to social involvement from all others.

\begin{figure}[htp]%
\centering
\subfloat[ $\C$ ]{{\includegraphics[width=0.29\textwidth, frame]{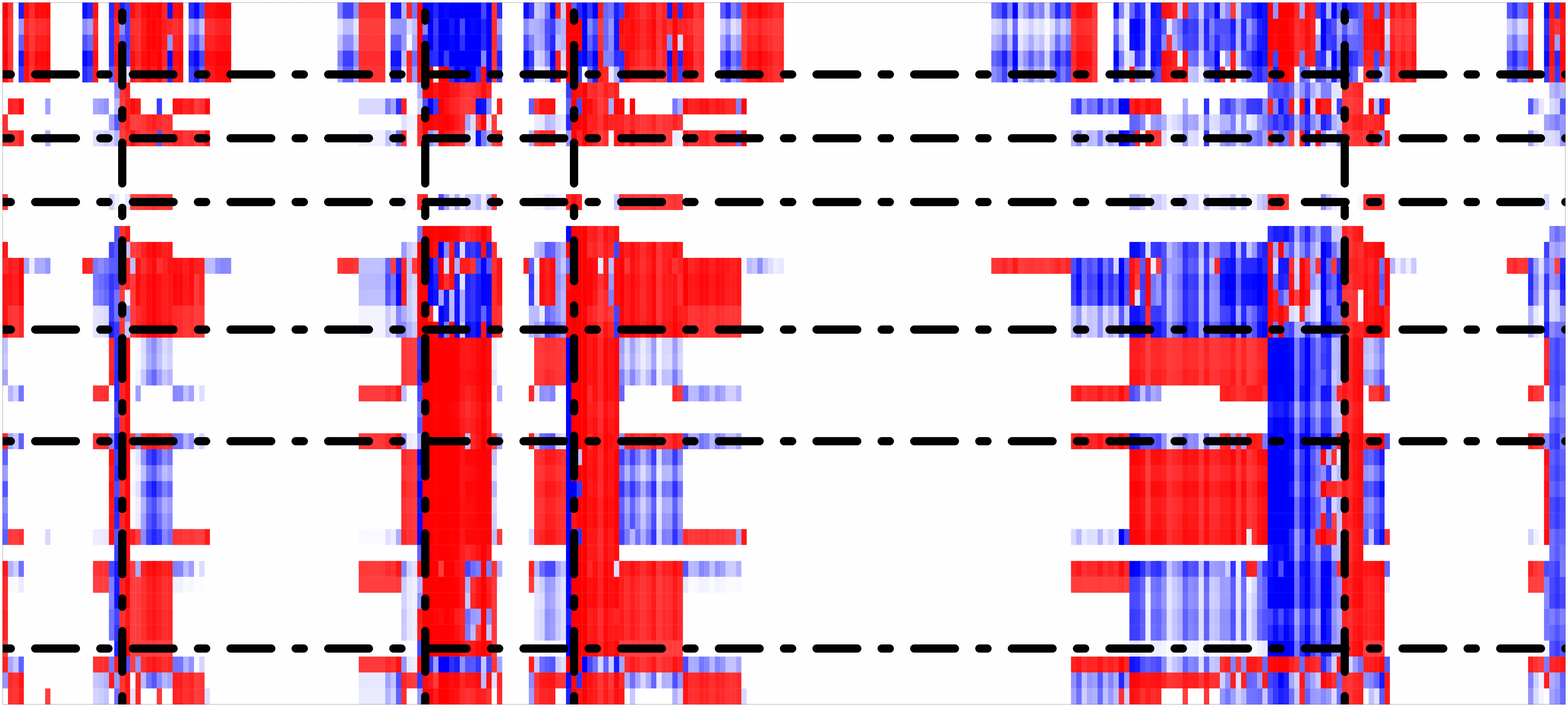} \quad }}
\subfloat[ $\C_1$]{{\includegraphics[width=0.29\textwidth, frame]{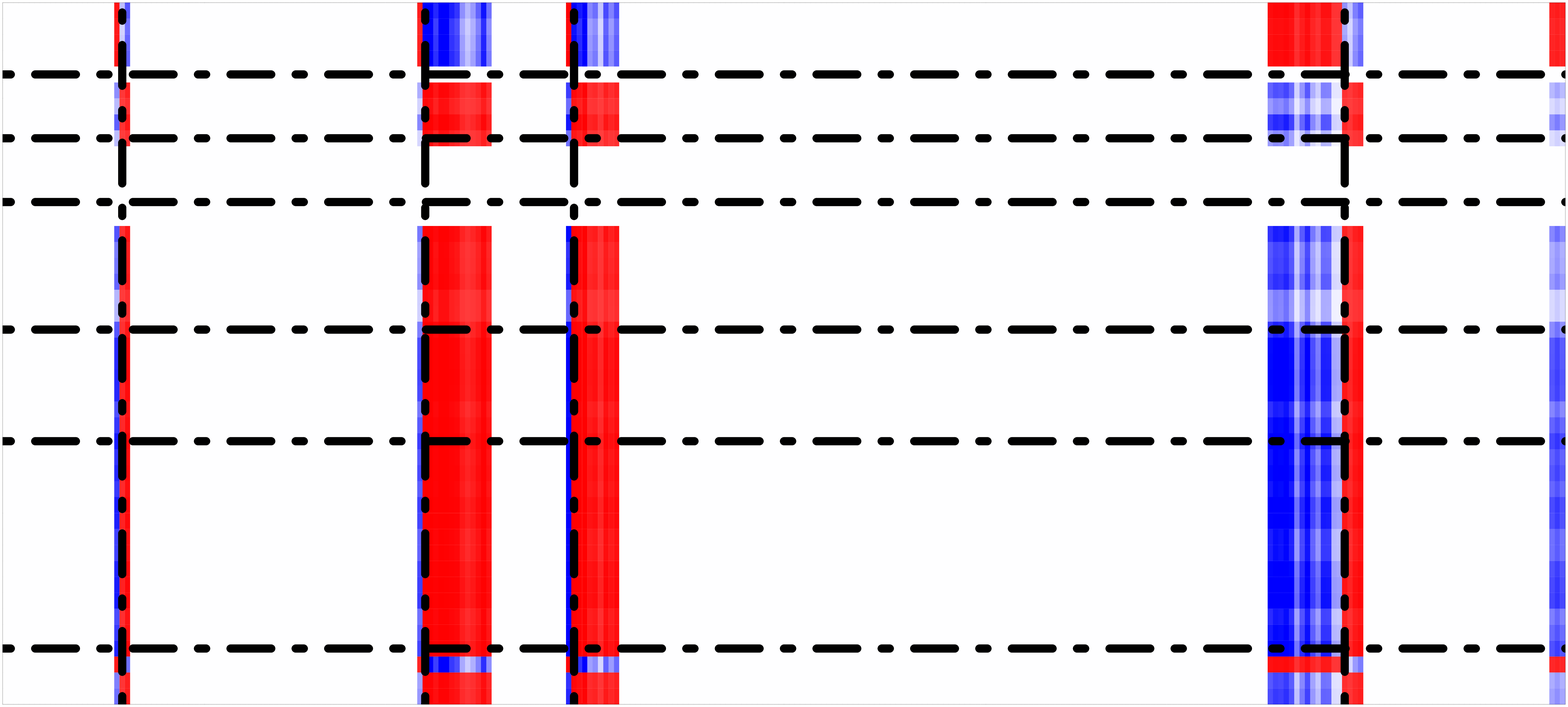} \quad }}
\subfloat[ $\C_2$]{{\includegraphics[width=0.29\textwidth, frame]{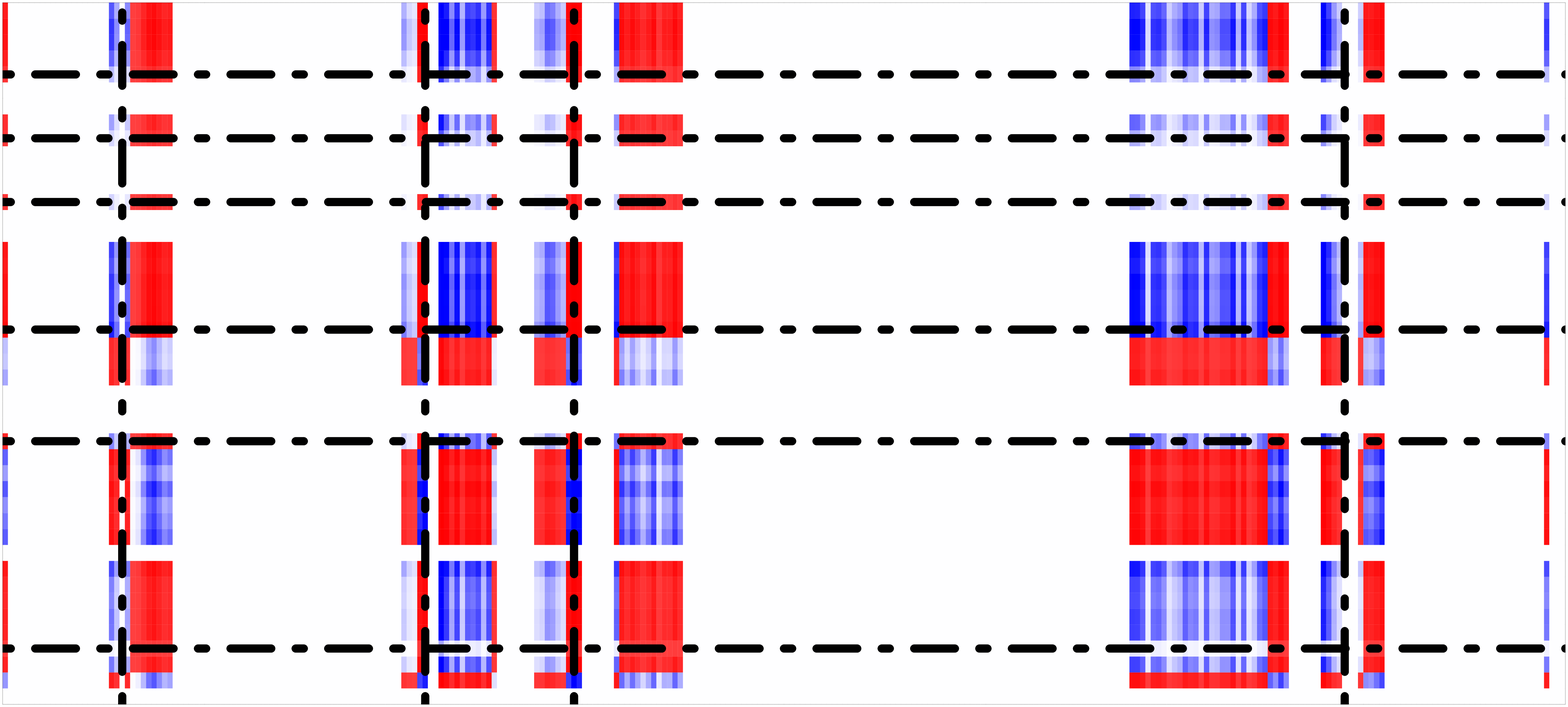}}} \\
\subfloat[ $\C_3$]{{\includegraphics[width=0.29\textwidth, frame]{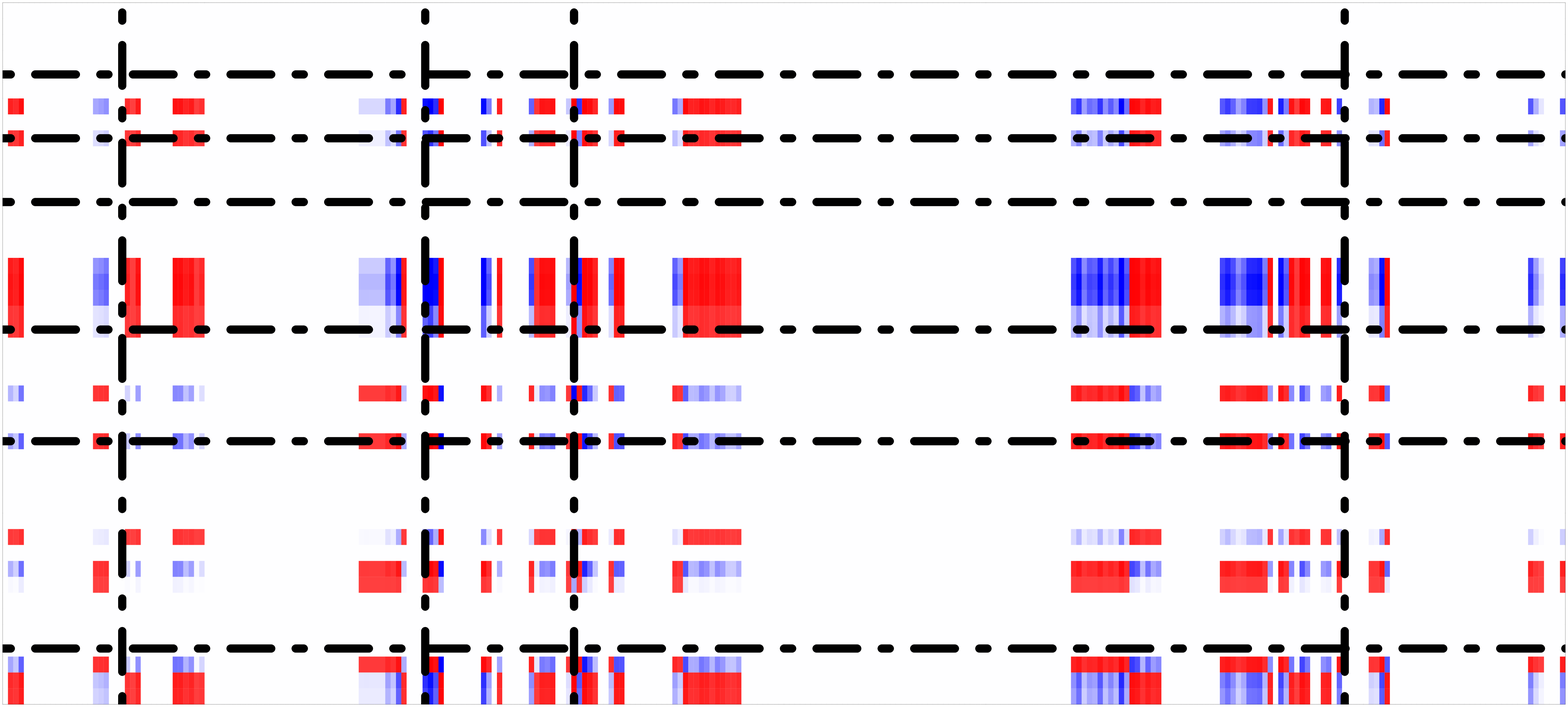} \quad }}
\subfloat[ $\C_4$]{{\includegraphics[width=0.29\textwidth, frame]{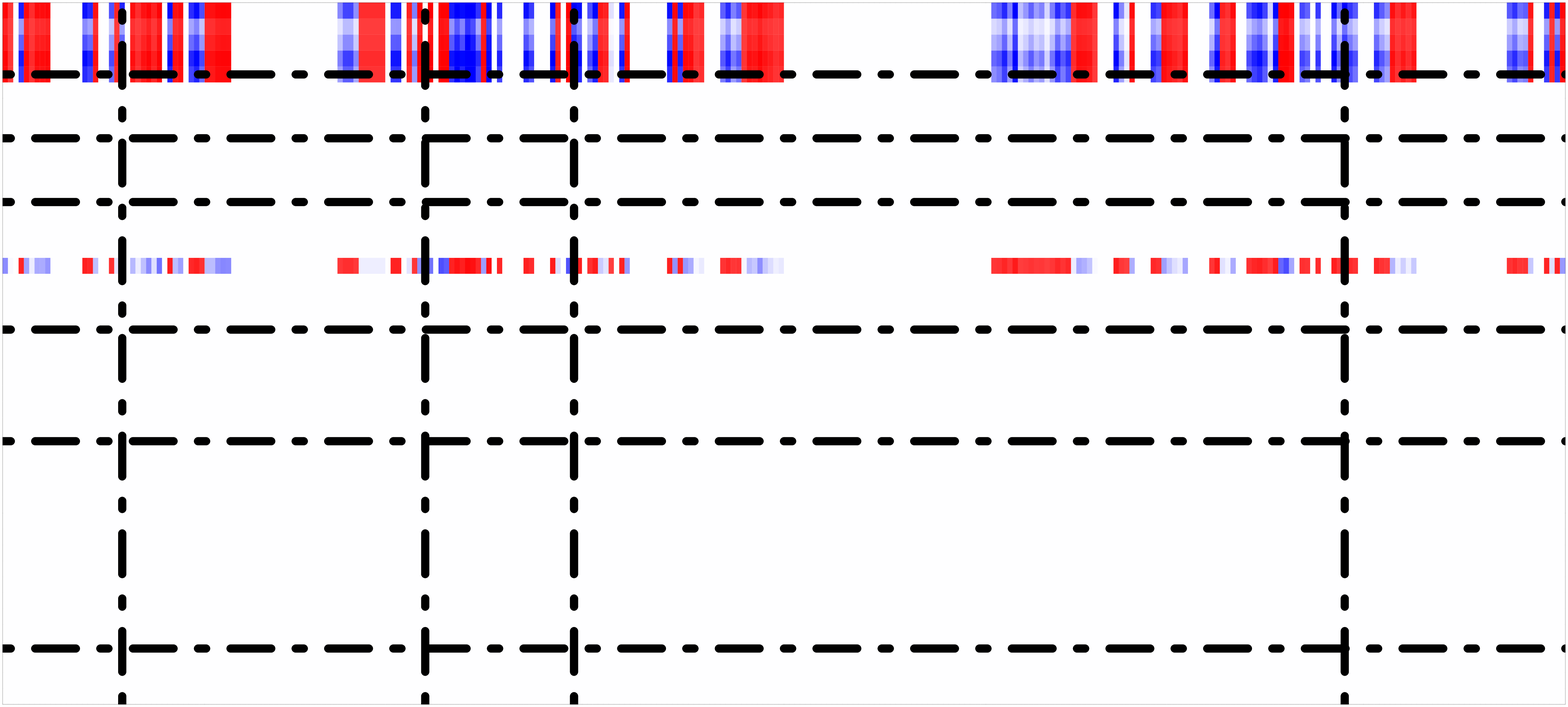} \quad }}	
\subfloat[ $\C_5$]{{\includegraphics[width=0.29\textwidth, frame]{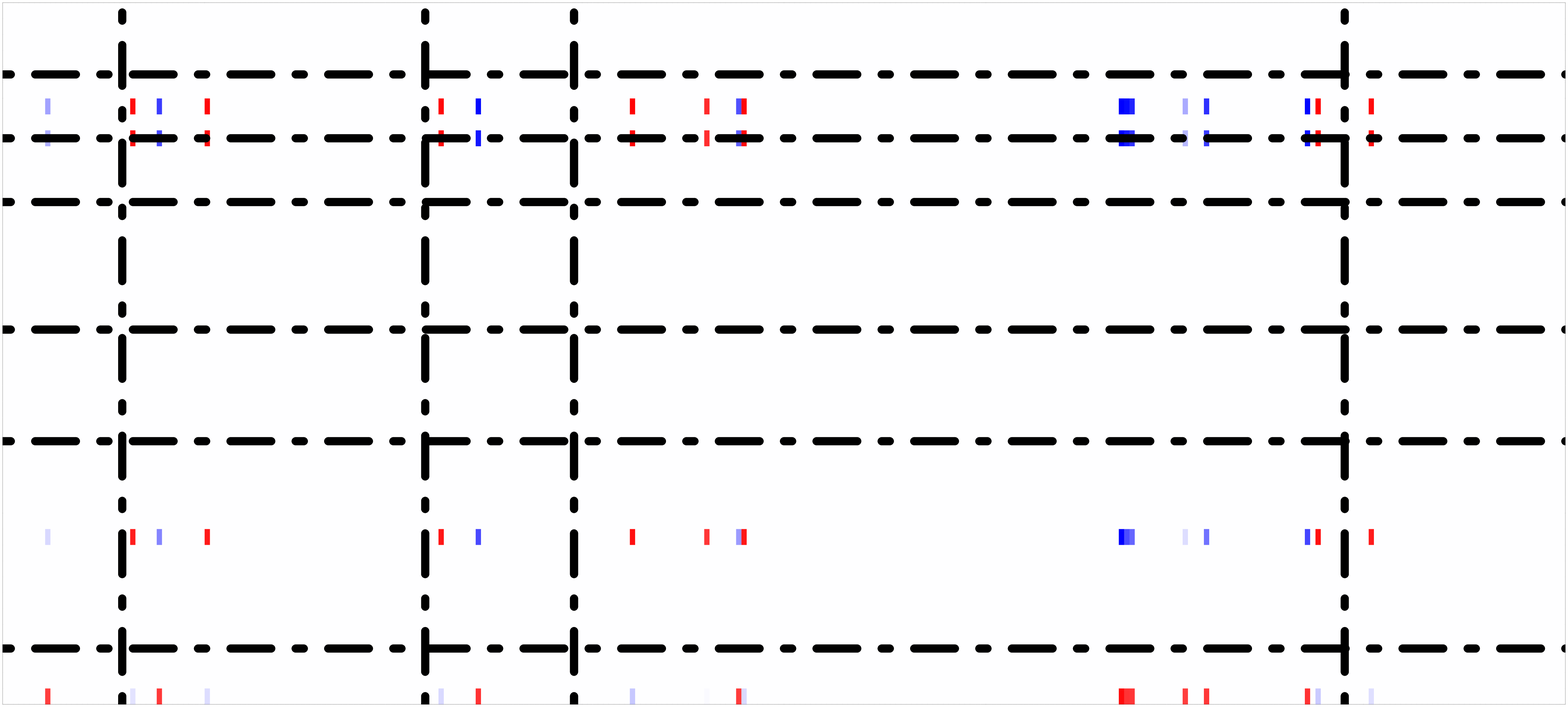}}}
\caption{Application -- LSOA Data: The sparse estimate of the coefficient matrix $\what{\C}$ with its unit-rank components using GOFAR(S). Horizontal lines separate the response into 7 categories  given by self-evaluation, fundamental daily activity, extended daily activity, medical condition, cognitive ability, sensation condition and social involvement (bottom to top). Vertical lines (left to right) separate the 294 predictors into five categories: namely, change in medical procedure since the last interview, daily activity, family status,  housing condition, and prior medical condition. 
	}
	\label{fig:mixedtype:lsoa-coef}%
\end{figure}

\subsection{Modeling of binary outcomes from CAL500}\label{subsec:cal500}
In the second application, we consider the Computer Audition Lab 500-song (CAL500) dataset \citep{turnbull2007towards} and apply the proposed procedure to explore the underlying associations. The data set consists of 68 audio signal characteristics  from signal processing as the predictor $\X$, and 174 annotations of songs by a human after listening as outcomes. The 
song features are mainly related to zero crossings, spectral centroid, spectral rolloff, spectral flux and Mel-Frequency Cepstral Coefficients (MFCC). On the other hand, the 174 binary outcomes from song annotations are categorized into emotions, genre, instrument, usage, vocals and song features. 
Some songs are annotated fewer  than 20 times. We merge the disjoint sets of  outcomes in a given category into one. After preprocessing, we are left with 107 binary outcomes ($\Y$). 
Since the underlying distribution of outcomes is Bernoulli, we model the song annotations using acoustic features and apply the proposed procedure to estimate the low-rank and sparse coefficient matrix. This allows us to find subsets of song features that affect only a subset of song annotations.

As in the LSOA data analysis, we compare the parameter estimates from  GOFAR(S),  GOFAR(P), mRRR and uGLM, and summarize the results from 100 replicates (80\% training and 20\% testing) in  Table \ref{tab:App:cal500:PECI}.  
All the rank-constrained approaches demonstrate better prediction error performance than the marginal modeling approach (uGLM), thus proving the merit of  the idea of using joint estimation to determine the underlying dependency. 
The prediction error performance of GOFAR(S),  GOFAR(P) and mRRR are comparable, with a slight edge to mRRR. This can be attributed to the fact that the underlying system is not sufficiently sparse (see the support of $\U$ and $\V$). We have already observed that the simulation results effectively demonstrate the usefulness of  GOFAR(S)/GOFAR(P) in both large and high-dimensional setups  where our underlying system is very sparse. Moreover, compared to the non-sparse model mRRR, GOFAR(S)/GOFAR(P) facilitates better interpretation via sparse singular vector estimates.

\begin{table}[H]
\centering
\caption{Application -- CAL500: Model evaluation (standard deviations are shown in parentheses) based on prediction error (PE), rank estimation $r$ and support recovery  \{\textsf{supp}($\U$) and \textsf{supp}($\V$)\}.}\label{tab:App:cal500:PECI}
\begingroup\scriptsize
\scalebox{1.0}{
\begin{tabular}{lcccc}
  \hline
Method & PE & r &  supp(U)\{\%\} & supp(V)\{\%\}  \\ 
  \hline
GOFAR(S) & 0.57(0.09) & 3.00(0.00) & 77(4) & 72(4) \\ 
  GOFAR(P) & 0.55(0.08) & 2.71(0.65) & 43(6) & 46(6) \\ 
  mRRR & 0.58(0.10) & 3.38(0.52) & 100(0) & 100(0) \\ 
  uGLM & 0.52(0.04) & 20.00(0.00) & 96(3) & 55(5) \\ 
   \hline
\end{tabular}
}
\endgroup
\end{table}

 \begin{figure}[htp]%
 \centering
\subfloat[ $\C$ ]{{\includegraphics[width=0.21\textwidth, frame]{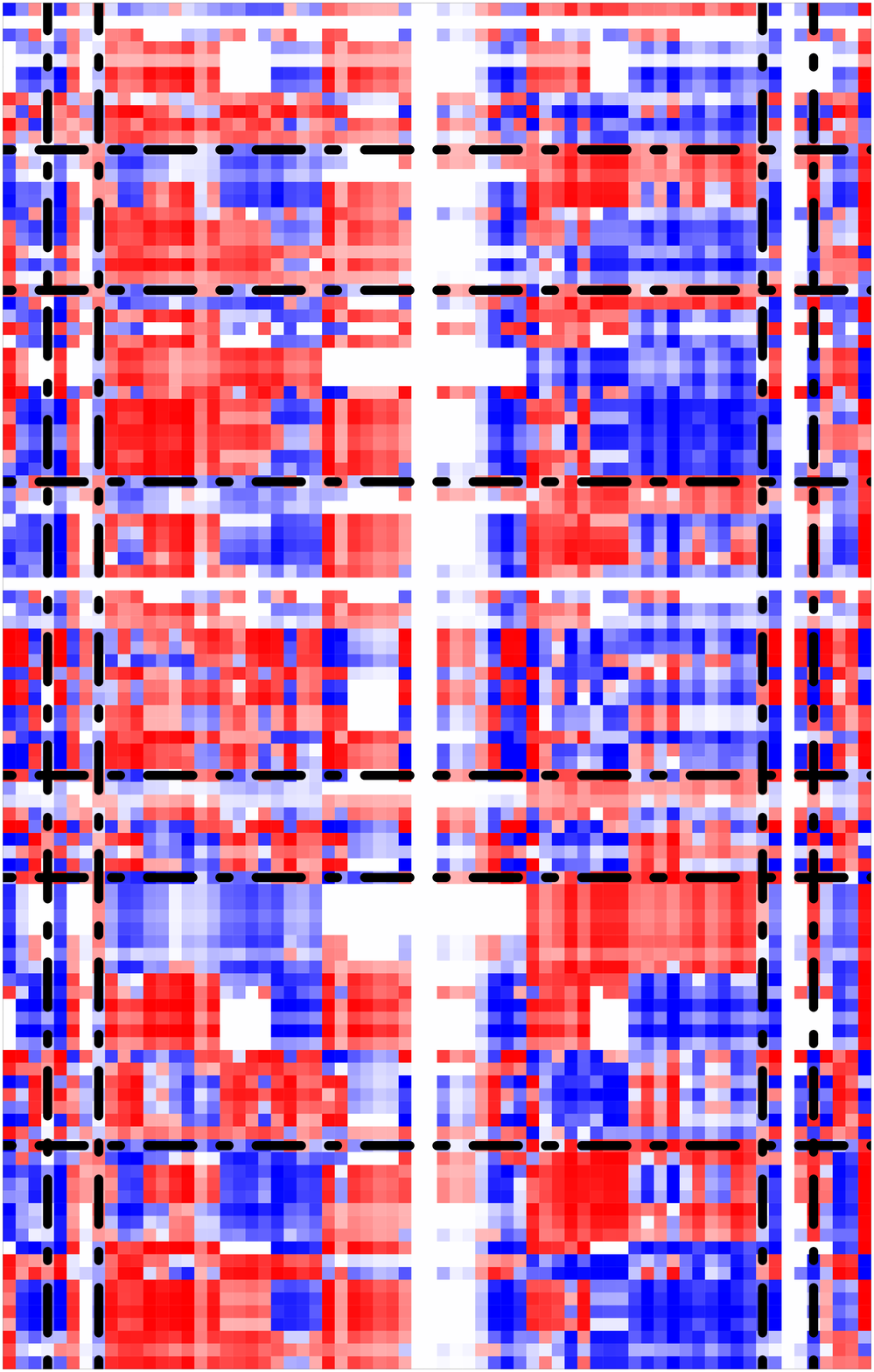} \quad }}
\subfloat[ $\C_1$]{{\includegraphics[width=0.21\textwidth, frame]{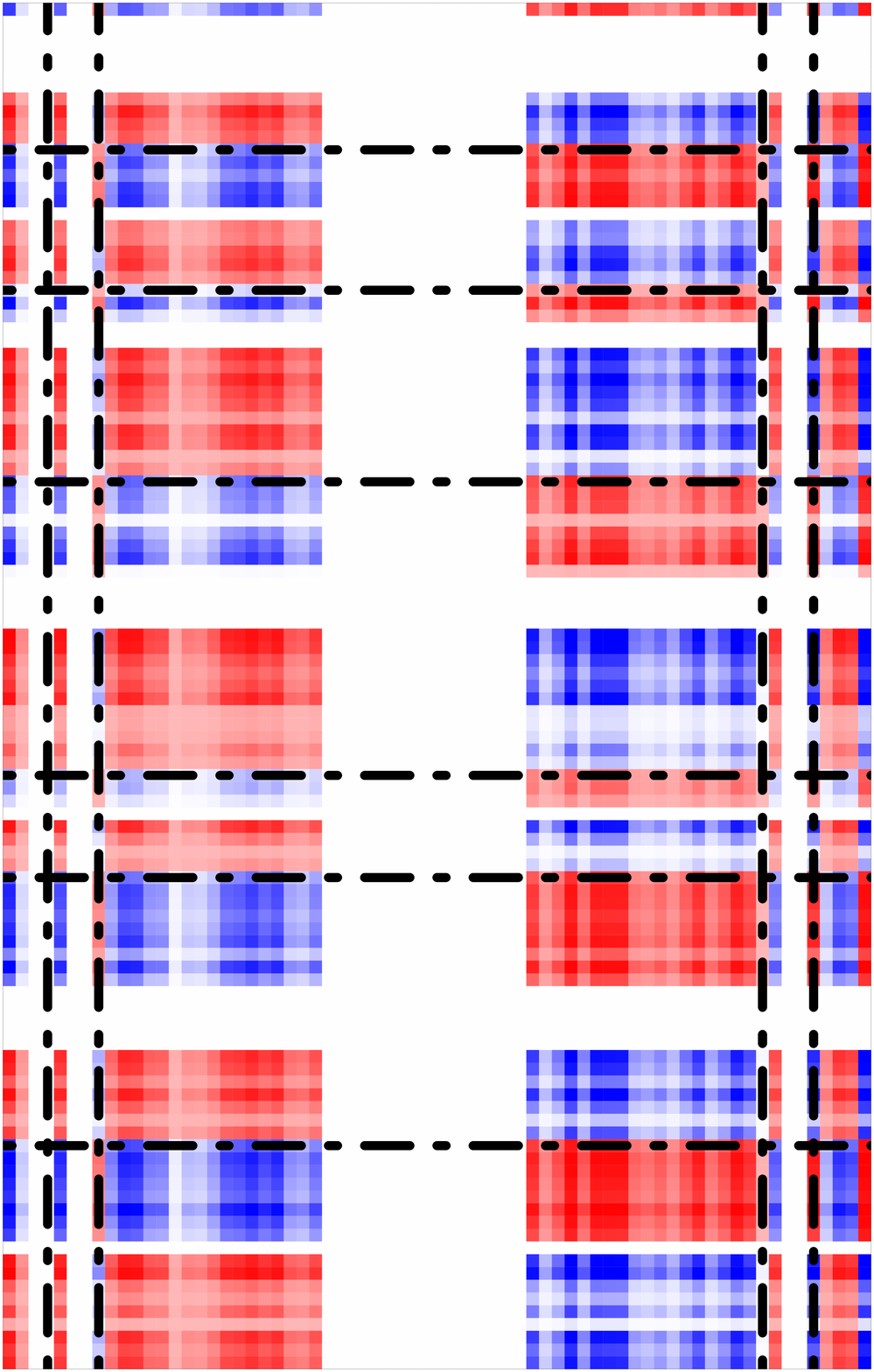} \quad }}
\subfloat[ $\C_2$]{{\includegraphics[width=0.21\textwidth, frame]{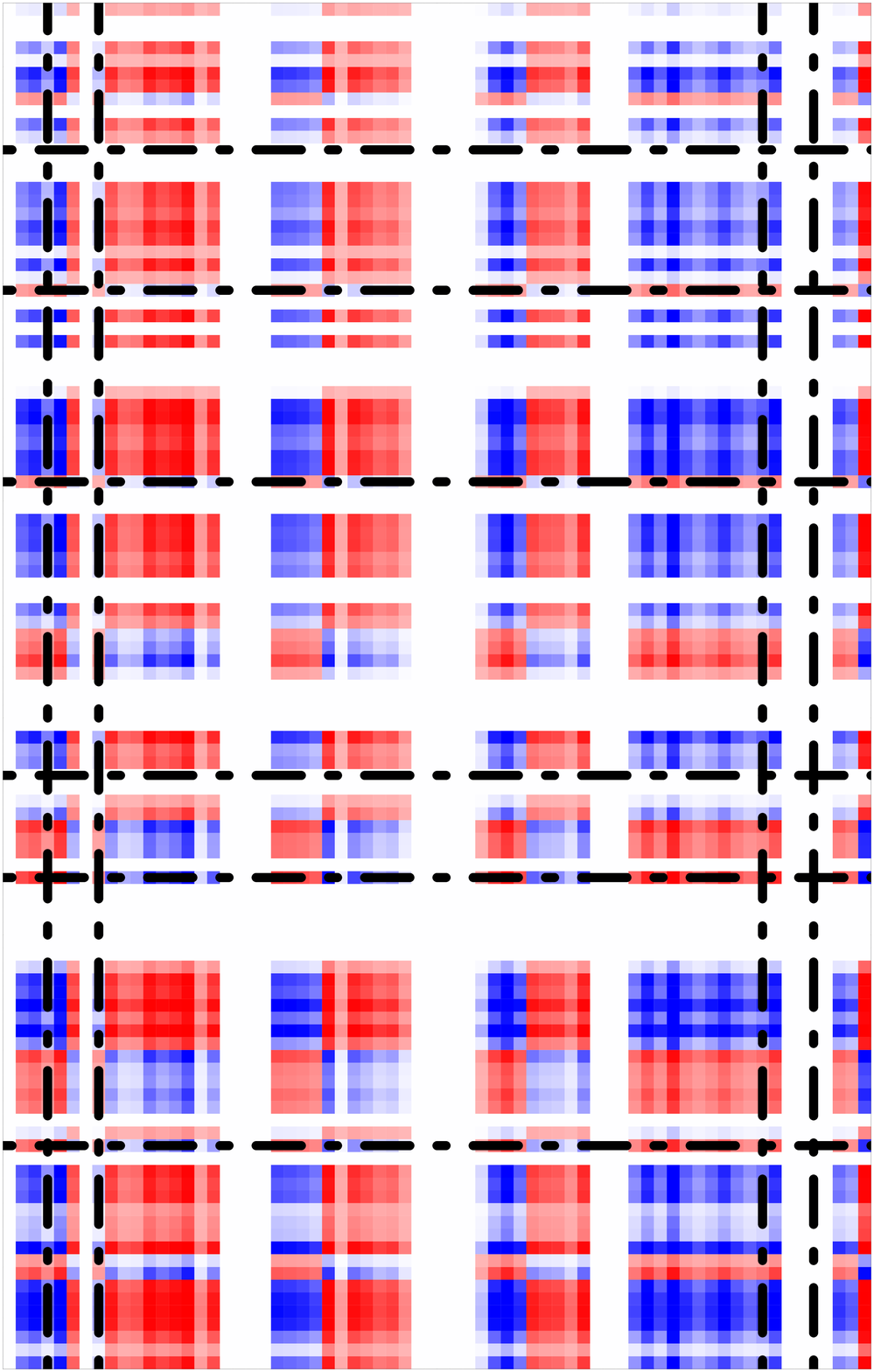} \quad }}
\subfloat[ $\C_3$]{{\includegraphics[width=0.21\textwidth, frame]{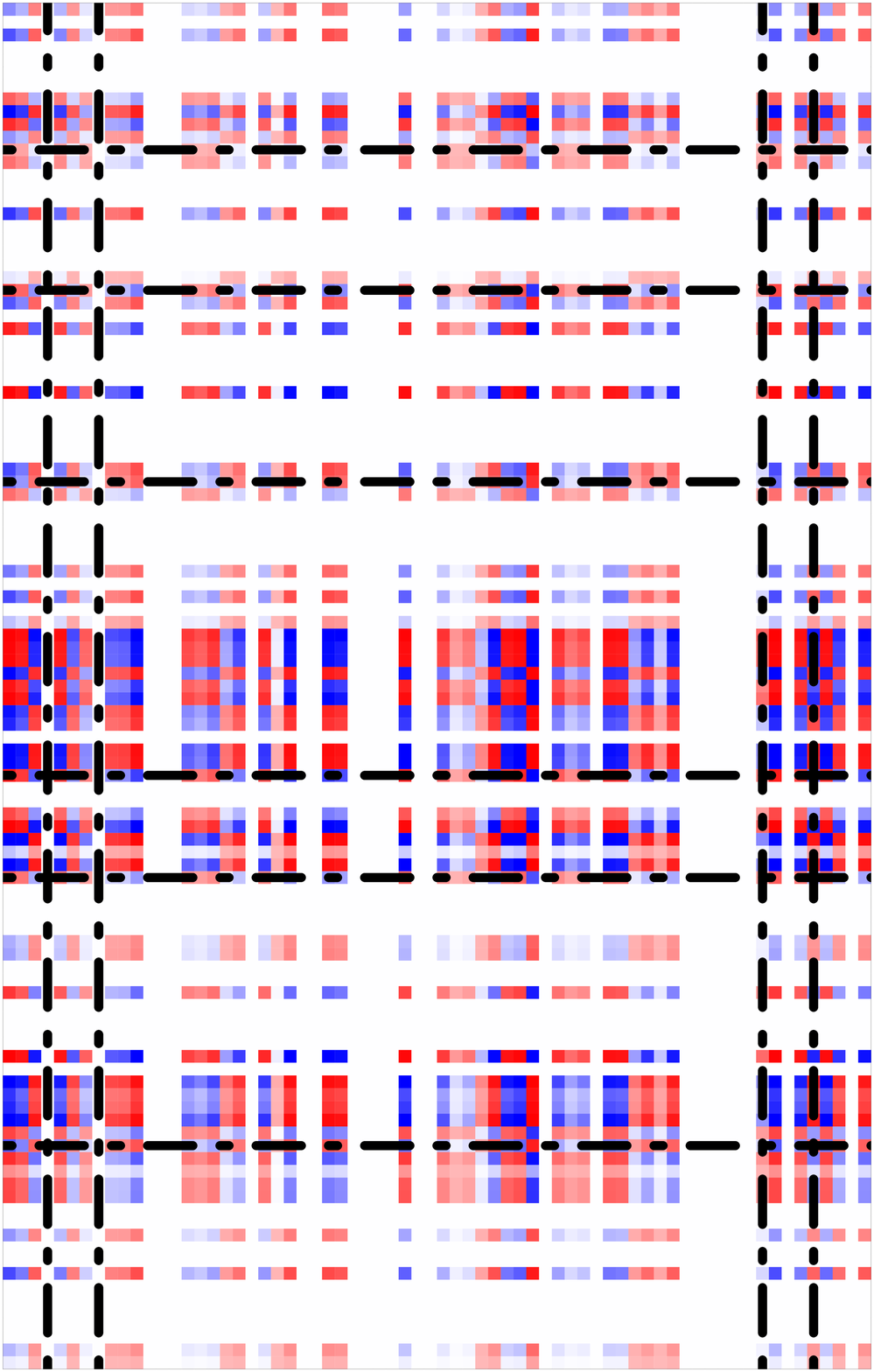} \quad }}		
	\caption{Application -- CAL500 Data: The sparse estimate of the coefficient matrix $\what{\C}$ with its unit-rank components using GOFAR(S).  Vertical lines (left to right) separate 68 predictors into five categories: namely,   spectral centroid, spectral flux, MFCC, spectral rolloff, and zero crossings. Horizontal lines separate the 102 response variables into seven groups: namely, emotion, genre, genre best, instrument, song, usage, and vocals (bottom to top). }
	\label{fig:mixedtype:cal500-coef}%
\end{figure}

Again, following the LSOA data analysis, because of the better support recovery of GOFAR(S), we apply this method to analyze the full data. 
Figure \ref{fig:mixedtype:cal500-coef} represents the low-rank and sparse coefficient matrix $\C$ and its  unit-rank components $\C_i$ for $i = 1,2,3$.  Through the row-wise sparsity of $\U$, the model discards  16 predictors overall, facilitating variable selection. 
Support of the unit-rank components is given by $\M{supp}(\U) = \{ 68\% , 74\%, 65\%\}$ and $\M{supp}(\V) = \{ 81\% , 72\%, 51\%\}$.  
Sparsity results in block cluster representation of the unit-rank components, so we interpret it accordingly.  
From the first unit-rank component, we clearly identify new subgroups in the song annotation category associated with the MFCC covariates.  The sign of the entries in the block matrix accordingly denotes the positive/negative associations. Among the MFCC covariates, we clearly find two separate subgroups.
Blocks resulting from the second unit-rank component estimate suggest the second set of covariates (mostly disjoint from first one) that are associated with a subset of song annotations. 
The third unit-rank component identifies features associated with a subgroup in the outcomes related to the instrument category.

In summary, as we demonstrated in two real-world examples, the proposed GOFAR is parsimonious and effective in recovering the underlying associations through the  sparse unit-rank components of the low-rank coefficient matrix.

\section{Discussion}\label{sec:dis}
In this article, we model the mixed type of outcomes via a multivariate extension of the GLM, with each response following a distribution in the exponential dispersion family. The model encodes the response-predictor dependency through an appealing co-sparse SVD of the nature parameter matrix. We develop two estimation procedures, i.e., a sequential method,  GOFAR(S) and a parallel method, GOFAR(P), to avoid the  notoriously difficult joint estimation alternative. 

There are many future research directions. Our model formulation \eqref{eq:glmmodel} is restricted to outcomes from the exponential dispersion family with canonical link; it would be interesting to consider more flexible link functions and other distributional families. Theoretically, we are interested in performing  non-asymptotic analysis to understand the finite sample behavior of the proposed estimators. Our approach handles missing entries in the response matrix using the same idea of matrix completion \citep{candes2009}; however, it may be fruitful to further explore the role that the type of missing entry plays in parameter estimation. Moreover, it is pressing to extend our method to handle missing entries in the predictor matrix. Finally, the proposed algorithms are still computationally intensive for large-scale problems; we will make the computation more scalable either by utilizing acceleration techniques in the current algorithms or by  developing path-following algorithms and stagewise learning procedures \citep{HeChen2018,chen2020statistically}.

\setcounter{section}{0}
\section{Supplementary Material}\label{sec:appendix}
\subsection{Exponential dispersion family}
\begin{table}[htp]
	\centering
	\caption{Some common distributions in the exponential dispersion family as specified in \citet{chenandluo2017}.}\label{tab:dist}
	\scalebox{0.8}{
		\begin{tabular}{lccccccc}
		\hline
		Distribution & Mean & Variance & $\theta$ & $\phi$ & $a(\phi)$ & $b(\theta)$ & $c(y;\phi)$\\
		\mbox{Bernoulli}($p$) & $p$ & $p(1-p)$ & $\log \{p(1-p)^{-1}\}$ & 1 & 1 & $\log(1+e^\theta)$ & 0 \\
		\mbox{Poisson}($\lambda$) & $\lambda$ & $\lambda$ & $\log\lambda$ & 1 & 1 & $e^\theta$ & $-\log y!$\\
		\mbox{Normal}($\mu$, $\sigma^2$) & $\mu$ & $\sigma^2$ & $\mu$ & $\sigma^2$ & $\phi$ & $\theta^2/2$ & $-(y^2\phi^{-1}+\log 2\pi)/2$ \\
		\hline
	\end{tabular}
}
\end{table}
\subsection{Simulation results}
\FloatBarrier 
\begin{table}[!h]
\centering
\parbox{1.00\textwidth}{\caption{Simulation: model evaluation based on 100 replications using various performance measures (standard deviations are shown in parentheses) in Setup I with Gaussian responses. \label{table:respTypeGaussianI}}} 
\begingroup\scriptsize
\scalebox{0.8}{
\begin{tabular}{llllllll}
  \hline
 & Er($\C$) & Er($\bs\Theta$) & FPR & FNR & R\% & r & time (s) \\  
   \hline 
& \multicolumn{7}{c}{ M\% = 0} \\
 \hline
GOFAR(S) & 16.95 (3.43) & 10.76 (1.99) & 0.08 (0.17) & 0.00 (0.00) & 0.00 (0.00) & 3.00 (0.00) & 6.75 (0.47) \\ 
  GOFAR(P) & 12.91 (3.24) & 9.38 (2.36) & 6.86 (4.54) & 0.00 (0.00) & 0.00 (0.00) & 3.00 (0.00) & 1.35 (0.14) \\ 
  mRRR & 34.77 (3.84) & 14.84 (1.72) & 100.00 (0.00) & 0.00 (0.00) & 0.00 (0.00) & 3.00 (0.00) & 27.73 (1.35) \\ 
  uGLM & 51.94 (6.67) & 29.85 (3.46) & 96.90 (0.87) & 0.00 (0.00) & 6.40 (0.85) & 25.72 (1.16) & 2.77 (0.10) \\ 
   \hline 
& \multicolumn{7}{c}{ M\% = 20} \\
GOFAR(S) & 21.50 (4.41) & 13.56 (2.28) & 0.22 (0.36) & 0.00 (0.00) & 0.00 (0.00) & 3.00 (0.00) & 7.76 (0.66) \\ 
  GOFAR(P) & 16.10 (4.32) & 11.73 (2.98) & 7.23 (4.49) & 0.00 (0.00) & 0.00 (0.00) & 3.00 (0.00) & 1.75 (0.17) \\ 
  mRRR & 43.37 (4.77) & 19.06 (2.18) & 100.00 (0.00) & 0.00 (0.00) & 0.00 (0.00) & 3.00 (0.00) & 30.99 (1.01) \\ 
  uGLM & 66.57 (8.48) & 38.56 (4.39) & 96.79 (1.05) & 0.00 (0.00) & 8.13 (1.21) & 25.78 (1.29) & 3.39 (0.12) \\ 
   \hline
\end{tabular}
}
\endgroup
\end{table}
 
\FloatBarrier 
\begin{table}[!h]
\centering
\parbox{1.00\textwidth}{\caption{Simulation: model evaluation based on 100 replications using various performance measures (standard deviations are shown in parentheses) in Setup I with Bernoulli responses. \label{table:respTypeBernoulliI}}} 
\begingroup\scriptsize
\scalebox{0.8}{
\begin{tabular}{llllllll}
  \hline
 & Er($\C$) & Er($\bs\Theta$) & FPR & FNR & R\% & r & time (s) \\  
   \hline 
& \multicolumn{7}{c}{ M\% = 0} \\
 \hline
GOFAR(S) & 42.50 (9.26) & 29.44 (4.55) & 0.73 (0.57) & 0.81 (1.31) & 0.00 (0.00) & 3.00 (0.00) & 37.71 (5.39) \\ 
  GOFAR(P) & 53.74 (14.53) & 35.83 (7.50) & 7.49 (4.99) & 1.80 (2.51) & 0.00 (0.00) & 3.00 (0.00) & 8.21 (0.96) \\ 
  mRRR & 183.39 (23.04) & 87.94 (13.54) & 100.00 (0.00) & 0.00 (0.00) & 0.00 (0.00) & 3.00 (0.00) & 29.19 (0.84) \\ 
  uGLM & 138.27 (8.15) & 94.83 (6.12) & 95.30 (1.64) & 0.00 (0.00) & 15.41 (2.15) & 25.18 (1.23) & 17.44 (0.92) \\ 
   \hline 
& \multicolumn{7}{c}{ M\% = 20} \\
GOFAR(S) & 52.53 (11.21) & 37.16 (6.37) & 1.45 (0.77) & 1.29 (1.53) & 0.00 (0.00) & 3.00 (0.00) & 54.95 (9.67) \\ 
  GOFAR(P) & 67.78 (16.89) & 47.65 (10.75) & 8.59 (4.99) & 3.73 (3.19) & 0.00 (0.00) & 3.00 (0.00) & 10.85 (1.13) \\ 
  mRRR & 253.03 (35.97) & 127.25 (21.24) & 100.00 (0.00) & 0.00 (0.00) & 0.00 (0.00) & 3.00 (0.00) & 28.64 (0.49) \\ 
  uGLM & 159.70 (9.73) & 113.73 (8.16) & 93.52 (2.67) & 0.67 (1.16) & 18.68 (2.77) & 25.05 (1.51) & 8.73 (0.40) \\ 
   \hline
\end{tabular}
}
\endgroup
\end{table}
 
\FloatBarrier 
\begin{table}[!h]
\centering
\parbox{1.00\textwidth}{\caption{Simulation: model evaluation based on 100 replications using various performance measures (standard deviations are shown in parentheses) in Setup I with Poisson responses. \label{table:respTypePoissonI}}} 
\begingroup\scriptsize
\scalebox{0.8}{
\begin{tabular}{llllllll}
  \hline
 & Er($\C$) & Er($\bs\Theta$) & FPR & FNR & R\% & r & time (s) \\  
   \hline 
& \multicolumn{7}{c}{ M\% = 0} \\
 \hline
GOFAR(S) & 4.83 (1.16) & 3.28 (0.62) & 0.33 (0.49) & 0.95 (1.30) & 0.00 (0.00) & 3.00 (0.00) & 390.72 (26.64) \\ 
  GOFAR(P) & 3.51 (0.74) & 2.48 (0.39) & 7.05 (4.57) & 0.00 (0.00) & 0.00 (0.00) & 3.00 (0.00) & 88.01 (5.52) \\ 
  mRRR & 17.78 (2.03) & 8.22 (0.74) & 100.00 (0.00) & 0.00 (0.00) & 3.06 (2.62) & 3.57 (0.56) & 31.66 (0.49) \\ 
  uGLM & 12.09 (1.28) & 7.31 (0.53) & 96.42 (1.13) & 0.00 (0.00) & 8.55 (1.02) & 25.44 (1.23) & 9.61 (0.29) \\ 
   \hline 
& \multicolumn{7}{c}{ M\% = 20} \\
GOFAR(S) & 5.80 (1.26) & 4.06 (0.80) & 0.49 (0.58) & 2.23 (2.48) & 0.00 (0.00) & 3.00 (0.00) & 462.17 (35.74) \\ 
  GOFAR(P) & 4.31 (0.91) & 3.02 (0.52) & 7.67 (4.62) & 0.00 (0.00) & 0.00 (0.00) & 3.00 (0.00) & 102.44 (5.46) \\ 
  mRRR & 21.62 (2.34) & 11.11 (0.96) & 100.00 (0.00) & 0.00 (0.00) & 2.95 (3.20) & 3.38 (0.70) & 31.66 (0.57) \\ 
  uGLM & 15.04 (1.36) & 9.33 (0.66) & 96.49 (1.27) & 0.00 (0.00) & 10.79 (1.26) & 25.51 (0.97) & 19.50 (0.94) \\ 
   \hline
\end{tabular}
}
\endgroup
\end{table}
 
\FloatBarrier 
\begin{table}[!h]
\centering
\parbox{1.00\textwidth}{\caption{Simulation: model evaluation based on 100 replications using various performance measures (standard deviations are shown in parentheses) in Setup I with Gaussian-Bernoulli responses. \label{table:respTypeGaussian-BernoulliI}}} 
\begingroup\scriptsize
\scalebox{0.8}{
\begin{tabular}{llllllll}
  \hline
 & Er($\C$) & Er($\bs\Theta$) & FPR & FNR & R\% & r & time (s) \\  
   \hline 
& \multicolumn{7}{c}{ M\% = 0} \\
 \hline
GOFAR(S) & 54.67 (6.12) & 43.89 (4.64) & 0.11 (0.21) & 0.00 (0.00) & 0.00 (0.00) & 3.00 (0.00) & 34.49 (6.16) \\ 
  GOFAR(P) & 34.20 (8.01) & 26.17 (5.83) & 7.75 (4.93) & 0.52 (1.12) & 0.00 (0.00) & 3.00 (0.00) & 7.30 (1.14) \\ 
  mRRR & 52.79 (6.39) & 24.60 (2.26) & 100.00 (0.00) & 0.00 (0.00) & 0.00 (0.00) & 3.00 (0.00) & 29.86 (0.99) \\ 
  uGLM & 101.97 (7.61) & 64.44 (4.59) & 94.73 (1.39) & 0.00 (0.00) & 10.73 (1.18) & 23.68 (1.56) & 10.17 (0.51) \\ 
  GOFAR(S,S) & 73.67 (17.81) & 56.31 (16.23) & 26.96 (4.84) & 0.46 (0.97) & 23.97 (27.81) & 6.69 (1.05) & 37.15 (12.35) \\ 
   \hline 
& \multicolumn{7}{c}{ M\% = 20} \\
GOFAR(S) & 63.34 (8.38) & 49.79 (5.89) & 0.35 (0.48) & 0.89 (1.19) & 0.00 (0.00) & 3.00 (0.00) & 40.76 (8.00) \\ 
  GOFAR(P) & 43.78 (13.38) & 33.90 (9.52) & 10.54 (6.29) & 2.38 (3.55) & 0.00 (0.00) & 3.00 (0.00) & 9.42 (1.91) \\ 
  mRRR & 65.97 (7.28) & 31.48 (2.78) & 100.00 (0.00) & 0.00 (0.00) & 0.00 (0.00) & 3.00 (0.00) & 30.76 (0.90) \\ 
  uGLM & 121.51 (8.94) & 79.40 (5.77) & 94.18 (1.68) & 0.00 (0.00) & 12.75 (1.63) & 23.69 (1.56) & 6.01 (0.26) \\ 
  GOFAR(S,S) & 98.05 (16.75) & 77.65 (14.57) & 28.15 (6.16) & 3.03 (3.43) & 24.83 (30.04) & 6.44 (1.14) & 59.19 (16.59) \\ 
   \hline
\end{tabular}
}
\endgroup
\end{table}
 
\FloatBarrier 
\begin{table}[!h]
\centering
\parbox{1.00\textwidth}{\caption{Simulation: model evaluation based on 100 replications using various performance measures (standard deviations are shown in parentheses) in Setup I with Gaussian-Poisson responses. \label{table:respTypeGaussian-PoissonI}}} 
\begingroup\scriptsize
\scalebox{0.8}{
\begin{tabular}{llllllll}
  \hline
 & Er($\C$) & Er($\bs\Theta$) & FPR & FNR & R\% & r & time (s) \\  
   \hline 
& \multicolumn{7}{c}{ M\% = 0} \\
 \hline
GOFAR(S) & 4.67 (0.93) & 2.88 (0.49) & 0.17 (0.33) & 0.00 (0.00) & 0.00 (0.00) & 3.00 (0.00) & 371.36 (14.90) \\ 
  GOFAR(P) & 3.06 (0.59) & 2.07 (0.34) & 4.71 (3.89) & 0.00 (0.00) & 0.00 (0.00) & 3.00 (0.00) & 85.02 (3.11) \\ 
  mRRR & 33.60 (1.88) & 28.48 (2.22) & 33.43 (0.00) & 67.24 (0.00) & 0.00 (0.00) & 1.00 (0.00) & 34.31 (0.58) \\ 
  uGLM & 11.77 (1.25) & 6.48 (0.53) & 94.98 (1.39) & 0.00 (0.00) & 8.62 (0.96) & 23.34 (1.44) & 6.08 (0.16) \\ 
  GOFAR(S,S) & 8.55 (2.46) & 6.13 (2.09) & 20.08 (1.87) & 0.00 (0.00) & 32.59 (36.76) & 5.36 (0.70) & 201.87 (14.32) \\ 
   \hline 
& \multicolumn{7}{c}{ M\% = 20} \\
GOFAR(S) & 5.77 (1.32) & 3.56 (0.69) & 0.25 (0.43) & 0.00 (0.00) & 0.00 (0.00) & 3.00 (0.00) & 378.17 (15.04) \\ 
  GOFAR(P) & 3.72 (0.73) & 2.53 (0.45) & 5.48 (4.48) & 0.00 (0.00) & 0.00 (0.00) & 3.00 (0.00) & 85.69 (2.88) \\ 
  mRRR & 36.40 (1.76) & 31.24 (2.35) & 33.43 (0.00) & 67.24 (0.00) & 0.00 (0.00) & 1.00 (0.00) & 33.77 (0.69) \\ 
  uGLM & 14.94 (1.45) & 8.44 (0.66) & 94.95 (1.39) & 0.00 (0.00) & 10.92 (1.30) & 23.71 (1.55) & 11.50 (0.73) \\ 
  GOFAR(S,S) & 8.47 (2.06) & 5.12 (1.37) & 21.41 (2.16) & 0.00 (0.00) & 36.35 (39.47) & 6.00 (0.00) & 231.92 (12.32) \\ 
   \hline
\end{tabular}
}
\endgroup
\end{table}

\FloatBarrier 
\begin{table}[!h]
\centering
\parbox{1.00\textwidth}{\caption{Simulation: model evaluation based on 100 replications using various performance measures (standard deviations are shown in parentheses) in Setup II with Gaussian responses. \label{table:respTypeGaussianII}}} 
\begingroup\scriptsize
\scalebox{0.8}{
\begin{tabular}{llllllll}
  \hline
 & Er($\C$) & Er($\bs\Theta$) & FPR & FNR & R\% & r & time (s) \\  
   \hline 
& \multicolumn{7}{c}{ M\% = 0} \\
 \hline
GOFAR(S) & 5.14 (1.69) & 8.42 (1.96) & 0.35 (0.21) & 0.00 (0.00) & 0.00 (0.00) & 3.00 (0.00) & 41.32 (3.09) \\ 
  GOFAR(P) & 5.37 (1.45) & 10.06 (2.02) & 3.35 (2.33) & 0.00 (0.00) & 0.00 (0.00) & 3.00 (0.00) & 10.67 (0.90) \\ 
  mRRR & 63.82 (7.30) & 116.25 (44.46) & 51.62 (24.36) & 47.35 (23.84) & 0.00 (0.00) & 1.54 (0.72) & 54.71 (1.34) \\ 
  uGLM & 25.85 (3.15) & 41.02 (4.18) & 87.90 (2.88) & 0.00 (0.00) & 7.82 (1.09) & 25.74 (1.14) & 7.39 (0.32) \\ 
   \hline 
& \multicolumn{7}{c}{ M\% = 20} \\
GOFAR(S) & 7.30 (1.65) & 12.23 (2.23) & 0.41 (0.27) & 0.00 (0.00) & 0.00 (0.00) & 3.00 (0.00) & 49.67 (3.57) \\ 
  GOFAR(P) & 6.80 (1.94) & 12.92 (3.17) & 4.06 (2.68) & 0.00 (0.00) & 0.00 (0.00) & 3.00 (0.00) & 13.20 (1.17) \\ 
  mRRR & 69.16 (5.35) & 145.05 (22.09) & 33.26 (0.00) & 65.31 (0.00) & 0.00 (0.00) & 1.00 (0.00) & 55.45 (1.21) \\ 
  uGLM & 32.78 (3.84) & 53.24 (5.27) & 85.38 (4.16) & 0.00 (0.00) & 10.02 (1.65) & 25.81 (1.35) & 5.46 (0.18) \\ 
   \hline
\end{tabular}
}
\endgroup
\end{table}




\subsection{Initialization}\label{subsec:initgsecure}

For a given rank $r$ and the offset term $\bO$,  consider the  optimization problem from mRRR \citep{chenandluo2017}
\begin{align}
(\widetilde{\C}, \widetilde{\bbeta},\widetilde{\Phi}) \equiv \argmin_{\C, \bbeta, \Phi}  \bbL(\bs\Theta,\Phi)  \qquad \M{s.t.} \qquad \M{rank}(\C) \leq r, \label{eq:gsfar:obj1}
\end{align}
where $\bs\Theta = \bs\Theta(\C, \bbeta, \bO)$. The joint estimation of the unknown parameters $(\C, \bbeta, \Phi)$ is nontrivial. To solve the problem, \citet{chenandluo2017} proposed an iterative procedure which proceeds via  $\C$-step, $\bbeta$-step and $\Phi$-step to update the parameters $\C$, $\bbeta$ and $\Phi$, respectively, until convergence.  
We have summarized the suggested procedure in  Algorithm \ref{alg:semrrr}. For convenience, let us denote the general class of problem by $\M{G-INIT}(\C,\bbeta,\Phi; \Y, \X,\bO,r)$. 

After solving the $\M{G-INIT}$ problem using Algorithm \ref{alg:semrrr}, we denote the parameter estimates by $\{\widetilde{\C}$,  $\widetilde{\bbeta}$, $\widetilde{\Phi} \}$. It is then trivial to retrieve the specific SVD decomposition of the coefficient matrix $\widetilde{\C}$ satisfying the orthogonality constraint $\widetilde{\U}\trans\X\trans\X\widetilde{\U}/n = \1$ and $\widetilde{\V}\trans\widetilde{\V} = \1$, where 
\begin{align}
\widetilde{\U} = [\tilde{\u}_1, \ldots, \tilde{\u}_r] , \,\,\,	\widetilde{\V} = [\tilde{\v}_1, \ldots,\tilde{\v}_r], \,\,\,
\widetilde{\D} = \M{diag}[\tilde{d}_1, \ldots, \tilde{d}_r], \label{eq:initval}
\end{align}
with $\widetilde{\C}_i = \tilde{d}_i\tilde{\u}_i \tilde{\v}_i\trans$. Thus, the $\M{G-INIT}(\C,\bbeta,\Phi; \Y, \X,\bO,r)$ problem outputs $\{\widetilde{\D}, \widetilde{\U},\widetilde{\V},\widetilde{\bbeta},\widetilde{\Phi} \}$. 

In summary, we solve $\M{G-INIT}(\C,\bbeta,\Phi; \Y, \X,\bO^{(k)},1)$ to initialize and construct weights in any $k$th step of the sequential approach \textsf{GOFAR(S)}. For the parallel approach \textsf{GOFAR(P)}, we simply solve $\M{G-INIT}(\C,\bbeta,\Phi; \Y, \X,\bO,r)$  to obtain an  initial estimate of the parameters that are used for  specifying the  offset terms and constructing weights.

\begin{algorithm}[ht]
	\caption{Initialization: $\M{G-INIT}(\C,\bbeta,\Phi; \Y, \X,\bO,r)$}
	\begin{algorithmic}\label{alg:semrrr}
	\STATE Given: $\X, \Y, \bO$ and desirable rank $r\geq 1$.
		\STATE Initialize: $\C^{(0)} = \0$, $\bbeta^{(0)}$, $\Phi^{(0)}$. 
		\REPEAT 
		\STATE (1) $\bC$-step: $\C^{(t+1)} = \mathbb{T}^{(r)} ( \C^{(t)} + \X\trans\{\Y - \bB'(\bs\Theta_c^{(t)})\}\Phi^{(t)-1})$ where $\bs\Theta_c^{(t)} = \bs\Theta(\C^{(t)},\bbeta^{(t)}, \bO)$, and $\mathbb{T}^{(r)}(\bM)$ extracts  $r$ SVD components of matrix $\bM$. 
		\vspace{0.1cm}
		\STATE (2) $\bbeta$-step:
		$\bbeta^{(t+1)} =  \bbeta^{(t)} + \Z\trans\{\Y - \bB'(\bs\Theta_{\beta}^{(t)} )\}\Phi^{(t)-1}$ where $\bs\Theta_{\beta}^{(t)} = \bs\Theta(\C^{(t+1)},\bbeta^{(t)}, \bO)$,\\
		\vspace{0.1cm}
		\STATE (3) $\Phi$-step:
		$\Phi^{(t+1)} = \arg\min_{\Phi}\sum_{i,k}  \bbL(\bs\Theta_{\Phi}^{(t)},\Phi) $ where $\bs\Theta_{\Phi}^{(t)} = \bs\Theta(\C^{(t+1)},\bbeta^{(t+1)}, \bO)$,\\
		\vspace{0.2cm}
		\STATE $t\gets t+1$.
		\UNTIL{convergence, \\
			\quad e.g., $\|[\C^{(t+1)} \,\,\bbeta^{(t+1)}] - [\C^{(t)} \,\,\bbeta^{(t)}]\|_F/ \|[\C^{(t)} \,\,\bbeta^{(t)}]\|_F \leq \epsilon$ with $\epsilon = 10^{-6}$.}
		\RETURN{$\widetilde{\C}$,  $\widetilde{\bbeta}$, $\widetilde{\Phi}$.}
	\end{algorithmic}
\end{algorithm}

\subsection{Analysis of the Convex Surrogate Function}\label{sec2:appendix_mmanalysis}
In the $\u$-step, for fixed $\{\v, \bbeta, \bPhi\}$ with $\v\trans\v = 1$, we rewrite the objective function \eqref{eq:gcure:objT} in terms of the product variable $\check{\u} = d\u$  to avoid the quadratic constraints, and conveniently denote it  as $F_{\lambda}(\check{\u},\v,\bbeta,\bPhi)$.
\begin{lemma}\label{th:lemma_majorize}
For fixed $\{\v, \bbeta, \bPhi\}$ with $\v\trans\v = 1$, $G_{\lambda}(\a;\check{\u}) \geq F_{\lambda}(\a,\v,\bbeta,\bPhi)$ for all $\a \in \mathbb{R}^p$ where $\check{\u} = d\u$ and the  scaling factor 
\begin{align}
s_u \geq \gamma_1 =\sup_{\check{\u} \in \mathbb{R}^p} \|\frac{\partial^2\bbL(\bs\Theta,\bPhi) }{\partial \check{\u} \partial \check{\u}\trans}\| = \sup_{\check{\u} \in \mathbb{R}^p}  \| \X\trans \sum_{k=1}^q v_k^2 \bs\zeta( \bs\Theta_{.k}(\check{\u}\v\trans, \bbeta ),a_k(\phi_k) ) \X \|, \label{eq:def_lux}
\end{align}
such that $\bs\zeta( \bs\Theta_{.k}(\check{\u}\v\trans, \bbeta ),a_k(\phi_k) )  = \M{diag}[ \bB_{.k}^{''}(\bs\Theta_{.k}(\check{\u}\v\trans, \bbeta ))]/a_k(\phi_k)$. 
\end{lemma}
\begin{proof}[Proof of Lemma \ref{th:lemma_majorize}: ]
For fixed $\{\v, \bbeta, \bPhi\}$ with $\v\trans\v = 1$, the continuously differentiable  negative log-likelihood function $\bbL(\bs\Theta(\check{\u}\v\trans),\bPhi)$ satisfy 
$$
\|\frac{\partial\bbL(\bs\Theta(\check{\u}\v\trans),\bPhi) }{\partial \check{\u}} - \frac{\partial\bbL(\bs\Theta(\a\v\trans),\bPhi) }{\partial \a}\| \leq \gamma_1  \|\a - \check{\u}\| \leq s_u   \|\a - \check{\u}\|,
$$
for any $ s_u \geq \gamma_1 =  \sup_{\check{\u} \in \mathbb{R}^p} \|\frac{\partial^2\bbL(\bs\Theta,\bPhi) }{\partial \check{\u} \partial \check{\u}\trans}\|$ (follows from the mean value theorem) given by 
\begin{align*}
\gamma_1= \sup_{\check{\u}}  \| \X\trans \sum_{k=1}^q v_k^2 \bs\zeta( \bs\Theta_{.k}(\check{\u}\v\trans, \bbeta ),a_k(\phi_k) ) \X \|,
\end{align*}
where $\bs\zeta( \bs\Theta_{.k}(\check{\u}\v\trans, \bbeta ),a_k(\phi_k) )  = \M{diag}[ \bB_{.k}^{''}(\bs\Theta_{.k}(\check{\u}\v\trans, \bbeta ))]/a_k(\phi_k)$. 
It is than trivial to show that  
\begin{align}
\bbL(\bs\Theta(\a\v\trans),\bPhi) \leq \bbL(\bs\Theta(\check{\u}\v\trans),\bPhi)   + (\a-\check{\u})\trans\frac{\partial\bbL(\bs\Theta,\bPhi) }{\partial \check{\u}}  +  \frac{s_u}{2} \|\a-\check{\u}\|_2^2. \label{eq:supp_res1}
\end{align}
On simplifying the surrogate function for fixed $\{\v, \bbeta, \bPhi\}$ with $\v\trans\v = 1$, we have 
\begin{align}
&G_{\lambda}(\a; \check{\u})  =  \bbL(\bs\Theta(\check{\u}\v\trans),\bPhi)   + \v\trans \bPhi^{-1} \{\bB^'(\bs\Theta(\check{\u}\v\trans )) -\Y \}\trans \X(\a-\check{\u})  +  \notag \\ 
& \qquad\frac{s_u}{2} \|\a-\check{\u}\|_2^2  +\rho(\a\v\trans;\W,\lambda)  \notag \\
&  =  \bbL(\bs\Theta(\check{\u}\v\trans),\bPhi)   + (\a-\check{\u})\trans\frac{\partial\bbL(\bs\Theta,\bPhi) }{\partial \check{\u}}  +  \frac{s_u}{2} \|\a-\check{\u}\|_2^2  +\rho(\a\v\trans;\W,\lambda).  \notag \\
&  \geq  \bbL(\bs\Theta(\a\v\trans),\bPhi)  + \rho(\a\v\trans;\W,\lambda)  =  F_{\lambda}(\a,\v,\bbeta,\bPhi), \notag
\end{align}
where last inequality follows from the result in \eqref{eq:supp_res1}. 
\end{proof}

Similarly, on extending Lemma \ref{th:lemma_majorize} for  $\v$-step, the surrogate function $H_{\lambda}(\b;\check{\v})$ majorizes the objective function  for 
\begin{align}
s_v \geq \gamma_2 =  \max_{1\leq k \leq q} \sup_{\check{\v} \in \mathbb{R}^q}   \|\u\trans \X\trans   \bs\zeta( \bs\Theta_{.k}(\u\check{\v}\trans, \bbeta ),a_k(\phi_k) ) \X\u\|, \label{eq:def_lvx} 
\end{align}
where $\bs\zeta( \bs\Theta_{.k}(\u\check{\v}\trans, \bbeta ),a_k(\phi_k))  = \M{diag}[ \bB_{.k}^{''}(\bs\Theta_{.k}(\u\check{\v}\trans, \bbeta ))]/a_k(\phi_k)$. 

Again, on extending Lemma \ref{th:lemma_majorize} for the $\bbeta$-step, the surrogate function $K(\balpha;\bbeta)$   majorizes the objective function  for  
\begin{align}
s_{\beta} \geq \gamma_3 =  \max_{1\leq k \leq q} \sup_{\bbeta}\| \Z\trans \bs\zeta(\bs \Theta_{.k}(\C, \bbeta ),a_k(\phi_k) )  \Z \|, \label{eq:def_lzx} 
\end{align}
where  $\bs\zeta( \bs\Theta_{.k}(\C, \bbeta ),a_k(\phi_k))  = \M{diag}[ \bB_{.k}^{''}(\bs\Theta_{.k}(\C, \bbeta ))]/a_k(\phi_k)$.

\subsection{Proof of Theorem \ref{th:convergence}}\label{sec2:motonodec}
	We show that Algorithm \ref{alg:spmrrr} admits desirable convergence properties. Let $\L^{(t)} = (d^{(t)},\u^{(t)}, \v^{(t)},  \bbeta^{(t)},\Phi^{(t)})$ be the parameter estimates in the $t$th iteration. 
	\paragraph{$\u$-step} Here parameters $\{\v^{(t)},  \bbeta^{(t)},\Phi^{(t)}\}$ are fixed. For convenience, we denote the objective function $F_{\lambda}( d^{(t)},\u^{(t)}, \v^{(t)},  \bbeta^{(t)},\Phi^{(t)}) $ by $F_{\lambda}( d^{(t)},\u^{(t)}) $. 
	Then, the unit-rank matrix $\C^{(t)} = d^{(t)} \u^{(t)} \v^{(t)}\trans$. For the  ease of presentation, we represent the natural parameter matrix as 
	\begin{align}
	\bs\Theta(\C^{(t)},\bbeta^{(t)}) = \bs\Theta(d^{(t)},\u^{(t)}, \v^{(t)},\bbeta^{(t)})  = \bO + \X\C^{(t)} + \Z \bbeta^{(t)}. \label{thetadef2}
	\end{align}
	Define $\check{\u} = d\u$ and $\check{\u}^{(t)} = d^{(t)}\u^{(t)}$. In the $\u$-step, the  unique and optimal solution minimizing the surrogate function $G_{\lambda}(\check{\u}; \check{\u}^{(t)})$ is given by $\check{\u}^{(t+1)} =  \td^{(t+1)} \u^{(t+1)} $. Using the result in Lemma \ref{th:lemma_majorize}, the convex surrogate function majorize the objective function in $\u$-step for the scaling factor $s_u$ \eqref{eq:def_lux},  i.e., 
	\begin{align*}
	F_{\lambda}( d, \u ) \leq G_{\lambda}(\check{\u}; \check{\u}^{(t)}) \qquad \forall \qquad \check{\u} \in \mathbb{R}^p.
	\end{align*}
	For the unique optimal solution $\check{\u}^{(t+1)}$, we have 
	\begin{align*}
	F_{\lambda}( \td^{(t+1)}, \u^{(t+1)}  ) \leq G_{\lambda}(\check{\u}^{(t+1)}; \check{\u}^{(t)}) \leq G_{\lambda}(\check{\u}^{(t)}; \check{\u}^{(t)}) = F_{\lambda}( d^{(t)},\u^{(t)}),
	\end{align*}
implies   $F_{\lambda}( \td^{(t+1)},\u^{(t+1)}, \v^{(t)},  \bbeta^{(t)},\Phi^{(t)}) \leq F_{\lambda}( d^{(t)},\u^{(t)}, \v^{(t)},  \bbeta^{(t)},\Phi^{(t)}) $.

		\paragraph{$\v$-step} Since parameters $\{\u^{(t+1)},  \bbeta^{(t)},\Phi^{(t)}\}$ are fixed, for convenience, denote the objective function $F_{\lambda}( d,\u^{(t)}, \v,  \bbeta^{(t)},\Phi^{(t)}) $ by $F_{\lambda}( d, \v) $. 
	Define $\check{\v}= d\v$ and $\check{\v}^{(t)} = \td^{(t+1)}\v^{(t)}$. In the $\v$-step, the unique and optimal solution minimizing the surrogate function $H_{\lambda}(\check{\v}; \check{\v}^{(t)})$ is given by $\check{\v}^{(t+1)} =  d^{(t+1)} \v^{(t+1)} $. Again, for the $\v$-step, we have  
		\begin{align*}
	F_{\lambda}( d, \v ) \leq  H_{\lambda}(\check{\v}; \check{\v}^{(t)}) \qquad \forall \qquad \check{\v} \in \mathbb{R}^q,
	\end{align*}
	for the suitable scaling factor $s_v$  \eqref{eq:def_lvx}. 
		For the unique optimal solution $\check{\v}^{(t+1)}$, we have 
	\begin{align*}
	F_{\lambda}( d^{(t+1)}, \v^{(t+1)}  ) \leq H_{\lambda}(\check{\v}^{(t+1)} ; \check{\v}^{(t)})\leq H_{\lambda}(\check{\v}^{(t)} ; \check{\v}^{(t)})= F_{\lambda}( \td^{(t+1)},\v^{(t)}),
	\end{align*}
implies   $F_{\lambda}( d^{(t+1)},\u^{(t+1)}, \v^{(t+1)},  \bbeta^{(t)},\Phi^{(t)}) \leq F_{\lambda}( \td^{(t+1)},\u^{(t+1)}, \v^{(t)},  \bbeta^{(t)},\Phi^{(t)}) $.

\paragraph{$\bbeta$-step}
Since parameters $\{\u^{(t+1)}, d^{(t+1)}, \v^{(t+1)},\Phi^{(t)}\}$ are fixed, for convenience, denote the objective function $F_{\lambda}( d^{(t+1)},\u^{(t+1)}, \v^{(t+1},  \bbeta^{(t)},\Phi^{(t)}) $ by $F_{\lambda}( \bbeta^{(t)}) $. 
In the $\bbeta$-step, the unique and optimal solution minimizing the surrogate function $K(\bbeta; \bbeta^{(t)})$ is given by $\bbeta^{(t+1)}$. Again, for the $\bbeta$-step, we have  
		\begin{align*}
	F_{\lambda}( \bbeta ) \leq  K(\bbeta; \bbeta^{(t)}) \qquad \forall \qquad \bbeta \in \mathbb{R}^{p_z \times q},
	\end{align*}
	for the suitable scaling factor $s_{\beta}$ (defined in Equation \eqref{eq:def_lzx}). 
			For the unique optimal solution $\bbeta^{(t+1)}$, we have 
	\begin{align*}
	F_{\lambda}( \bbeta^{(t+1)} ) \leq K(\bbeta^{(t+1)}; \bbeta^{(t)}) \leq K(\bbeta^{(t)}; \bbeta^{(t)})= F_{\lambda}( \bbeta^{(t+1)}),
	\end{align*}
implies   $F_{\lambda}( d^{(t+1)},\u^{(t+1)}, \v^{(t+1)},  \bbeta^{(t+1)},\Phi^{(t)}) \leq F_{\lambda}( d^{(t+1)},\u^{(t+1)}, \v^{(t+1)},  \bbeta^{(t)},\Phi^{(t)})$.

	Finally, the unknown dispersion parameters are estimated based on maximizing the log-likelihood, so it is guaranteed to have a non-increasing objective function. Thus, on adding the  results from the $\u$-step, $\v$-step and $\bbeta$-step, the proof of Theorem \ref{th:convergence} easily follows.

\subsection{Proof of Theorem \ref{Sec2:TH:localminima}}\label{supp:asyproof}
Using the SVD in \eqref{sec2:eq:cstar}, we define the set $\Omega_k$ as
	$$
	\Omega_k = \left\{(\u_k, \v_k, \bbeta): \u_k \in \mathbb{R}^p \M{ and  } \v_k \in  \mathbb{R}^q  \M{ with } v_{\ell_k k} = 1 , \bbeta \in \bbR^{p_z \times q} \right\}.
	$$
	Then, $(\what{\u}_k,\what{\v}_k, \what{\bbeta}) \in \Omega_k$. 
	We first prove the result for $k=1$. Consider a neighborhood of ($\u_1^*, \v_1^*, \bbeta^*$) with radius $h>0$,
	\begin{align*}
	\mathcal{N}(\u_1^*, \v_1^*,\bbeta^*,h) = &   \{(\u_1^*+\a/\sqrt{n})(\v_1^*+\b/\sqrt{n})\trans , \bbeta^* + \mcA/\sqrt{n} \} ;\\ 
	&\M{s.t.} \quad \|\bGamma^{1/2}\a\| \leq h, \a \in\mathbb{R}^p, \b \in \mathbb{R}^q, \|\b\|\leq h, b_{\ell_1}=0, \|\mcA\|\leq h .
	\end{align*}
	We claim that for any $\epsilon > 0$, there exists a large enough $h$ such that 
	\begin{align}
	P &\left\{ \underset{\substack{\|\Gamma^{1/2}\a\| = \|\b\| = \\ \|\mcA\|  = h}}{\M{inf}}  F_1^{(n)}(\what{\u}_1,\what{\v}_1, \what{\bbeta} )  >  F_1^{(n)}(\u_1^*,\v_1^*,\bbeta^*)  \right\}  \geq 1-\epsilon \label{sec2:the:consist1}
	\end{align}
	where $\what{\u}_1 = \u_1^*+\a/\sqrt{n}$, $\what{\v}_1 = \v_1^*+\b/\sqrt{n}$, and $\what{\bbeta} = \bbeta^*+\mcA/\sqrt{n}$. 
	The claim implies that with probability at least $1-\epsilon$,  there exists a local minimum $(\what{\u}_1,\what{\v}_1,\what{\bbeta})$ in the interior of $\mathcal{N}(\u_1^*,\v_1^*,\bbeta^*,h)$, resulting in  $\|\what{\u}_1 - \u_1^* \| = O_p(n^{-1/2})$, $\|\what{\v}_{1 } - \v_{1}^* \| = O_p(n^{-1/2})$ and $\|\what{\bbeta} - \bbeta^* \| = O_p(n^{-1/2})$.
	
	Now, let us assume that  $\what{\C}_1 = \what{\u}_1 \what{\v}_1\trans$ and $\C_1^* = \u_1^*\v_1^*\trans$, and write the  true and estimated natural parameter matrix by  $\bs\Theta_1^* = \bO_1 + \X\C_1^* + \Z\bbeta^*$ and  $\what{\bs\Theta}_1 = \bO_1 + \X\what{\C}_1 + \Z\what{\bbeta}$, respectively. 
To prove the result in \eqref{sec2:the:consist1}, we define
	\begin{align}
	\Psi_1^{(n)}(\a,\b,\mcA)  & = F_1^{(n)}(\u_1^*+\a/\sqrt{n},\v_1^*+\b/\sqrt{n}, \bbeta^*+\mcA/\sqrt{n})  -  F_1^{(n)}(\u_1^*,\v_1^*,\bbeta^*) \notag \\
	& = T_1 + T_2 + T_3, \label{sec2:psifun} 
	\end{align}
	where 
	\begin{equation}
	\left .
	\begin{aligned}
	T_1 & = - \Tr(\Y\trans\{\what{\bs\Theta}_1 - \bs\Theta_1^* \} )+ \Tr(\J\trans \{ \bB(\what{\bs\Theta}_1) - \bB(\bs\Theta_1^*) \} )\\
	T_2 & = \alpha\lambda_1^{(n)}  \left\{\| \W_1 \circ  (\u_1^*+\a/\sqrt{n}) (\v_1^*+\b/\sqrt{n})\trans\|_1 - \|\W_1 \circ  \u_1^*\v_1\strans \|_1 \right\}, \\
	T_3 & =(1- \alpha)\lambda_1^{(n)} \{ \|(\u_1^*+\a/\sqrt{n}) (\v_1^*+\b/\sqrt{n})\trans \|_F^2-  \|\u_1^*\v_1\strans \|_F^2 \}.
	\end{aligned}
	\right \} . \label{sec2:TTT}
	\end{equation}
	Using the details  of the proof of Theorem 5.2 in \citet{mishra2017sequential}, it can be verified that the terms $T_2$ and $T_3$ are of $O(h)$. 
	Now, we focus on simplifying $T_1$. Using assumption \textbf{A4} (strictly convex), we have 
	$$
	\Tr(\J\trans\{\bB(\what{\bs\Theta}_1) - \bB(\bs\Theta_1^*)\})\,\, \geq \,\, \Tr( \{\bB^'(\bs\Theta_1^*)\}\trans \{\what{\bs\Theta}_1 - \bs\Theta_1^* \})+ {\gammaL \over 2} \|\what{\bs\Theta}_1 - \bs\Theta_1^* \|_F^2,
	$$
	and write 
		\begin{align}
	T_1 & \geq - \Tr( \Y\trans \{\what{\bs\Theta}_1 - \bs\Theta_1^*\} )   + \Tr( \{ \bB^'(\bs\Theta_1^*)\}\trans \{\what{\bs\Theta}_1 - \bs\Theta_1^*\} ) + {\gammaL \over 2} \|\what{\bs\Theta}_1 - \bs\Theta_1^* \|_F^2 \notag \\
	& = - T_{11} + T_{12} , \label{eq:consistencyT1}
	\end{align}
	where 
	\begin{align*}
	T_{11} & =\Tr( \E\trans \{\bs\Theta_1^* -\what{\bs\Theta}_1\} ) \\
	T_{12} &=\Tr( \{\bB^'(\bs\Theta_1^*) -\bB^'(\bs\Theta^*)\}\trans \{\what{\bs\Theta}_1 - \bs\Theta_1^* \} )+ {\gammaL \over 2} \|\what{\bs\Theta}_1 - \bs\Theta_1^* \|_F^2.
	\end{align*}
{ Then, the bound  $T_{11} \leq  \| \bP_{\widetilde{\X}} \E \| \{\M{rank}(\C^*) + \M{rank}(\what{\C}) + \M{rank}(\Z\bbeta^*) \} \| \bs\Theta_1^* -\what{\bs\Theta}_1\|_F$ implies $T_{11}$ is of $O(h)$ where $\bP_{\widetilde{\X}}$ is the projection matrix of $\widetilde{\X} = [ \Z \,\, \X]$.}
On combining the results obtained so far, we get 
	\begin{align*}
	\Psi_1^{(n)}(\a,\b,\mcA)  &\geq T_{12} + O(h) + O_p(1/\sqrt{n}).
	\end{align*}

	Now, we shift our focus to  simplifying $T_{12}$. Again, using assumption \textbf{A4}, we have  
	$$
	\Tr( \{\bB^'(\bs\Theta_1^*) -\bB^'(\bs\Theta^*) \}\trans \{\what{\bs\Theta}_1 - \bs\Theta_1^*\} )\geq \gammaL\Tr( \{\bs\Theta_1^* -\bs\Theta^* \}\trans \{\what{\bs\Theta}_1 - \bs\Theta_1^*\} ),
	$$
	resulting in 
	\begin{align*}
	T_{12} \geq   \gammaL \Tr(\{\bs\Theta_1^* -\bs\Theta^* \}\trans\{\what{\bs\Theta}_1 - \bs\Theta_1^*\} ) + {\gammaL  \over 2} \|\what{\bs\Theta}_1 - \bs\Theta_1^* \|_F^2 .
	\end{align*}
We simplify $T_{12}$ as per the  choice of the estimation procedure, i.e., GOFAR(S) and GOFAR(P). 

\paragraph{CASE - GOFAR(S)}
On substituting the $\bs\Theta^* = \X\C^* + \Z\bbeta^*, \bs\Theta_1^* = \X\C_1^* + \Z\bbeta^*$ and $\what{\bs\Theta}_1 = \X\what{\C}_1 + \Z\what{\bbeta}$, we get 
	\begin{align}
	T_{12} \geq   T_{12}^a + T_{12}^b  , \label{eq:t12consis} 
	\end{align}
	where $T_{12}^a = \gammaL \Tr(\{-\X \C_{-1}^*\}\trans \{ {\Z\mcA \over \sqrt{n}} +  { \X \over \sqrt{n} } (\u_1^* \b\trans + \a\v_1\strans + {\a\b\trans \over \sqrt{n}} ) \} )$ and $T_{12}^b =  {\gammaL  \over 2n} \| \Z\mcA +  \X (\u_1^* \b\trans + \a\v_1\strans)\|_F^2$
	with $\C_{-1}^* = \sum_{l>1}^{r^*} \C_{l}^* $. Under  assumption \textbf{A1} and the orthogonal decomposition  defined in equation \eqref{sec2:eq:cstar}, we simplify  the term on the  right-hand side  of \eqref{eq:t12consis}  as 
	\begin{align*}
	& T_{12}^b  = {1 \over 2n} \| \Z\mcA +  \X (\u_1^* \b\trans + \a\v_1\strans)\|_F^2 = {1 \over 2n} (\| \Z\mcA\|_F^2 +  \|\X (\u_1^* \b\trans + \a\v_1\strans)\|_F^2),\\
	& T_{12}^a  =\Tr( \{-\X \C_{-1}^* \}\trans { \X \over \sqrt{n} } (\u_1^* \b\trans + \a\v_1\strans + {\a\b\trans \over \sqrt{n}} )  )= -\sum_{l>1}^{r^*} \a\trans \bGamma_1\u_l^* \b\trans \v_l^* ,\\
	&\Tr(\{-\X \C_{-1}^*\}\trans  {\Z\mcA \over \sqrt{n}} )= 0 
	\end{align*}
	Thus, we have 
	\begin{align*}
	T_{12} \geq   \gammaL \Big( {1 \over 2n}\| \Z\mcA\|_F^2 +  {1 \over 2n}\|\X (\u_1^* \b\trans + \a\v_1\strans)\|_F^2  -\sum_{l>1}^{r^*} \a\trans \bGamma_1\u_l^* \b\trans \v_l^* \Big).
	\end{align*}
	Under the assumptions \textbf{A1 - A4}, the proof of Theorem 5.3  in \citet{mishra2017sequential} suggests that the term 
	$$
	{1 \over 2n}\|\X (\u_1^* \b\trans + \a\v_1\strans)\|_F^2  -\sum_{l>1}^{r^*} \a\trans \bGamma_1\u_l^* \b\trans \v_l^* \geq 0, 
	$$
	i.e., it is positive semidefinite  and of order $O(h^2)$. Also,  $ {1 \over 2n}\| \Z\mcA\|_F^2$ is of order $O(h^2)$. From this, we conclude that the quadratic terms $	T_{12} $ involving  $\a$, $\b$ and $\mcA$ are of $O(h^2)$ and positive. It dominate the other terms of order $O(h)$ for a sufficiently large $h$. Hence, $\Psi_1^{(n)}(\a,\b,\mcA) \geq 0$.

	Now, we extend the result for the estimate of the  $k$th unit-rank component, i.e.,  ($\what{\u}_k,\what{\v}_k$). For $l=1,\ldots,k-1$, we have $\|\what{\u}_l - \u_l^* \| = O_p(n^{-1/2})$ and $\|\what{\v}_{l} - \v_{l}^* \| = O_p(n^{-1/2})$. Define	
	\begin{align*}
	\mathcal{N}(\C_k^*,\bbeta^*,h) = &   \{(\u_k^*+\a/\sqrt{n})(\v_k^*+\b/\sqrt{n})\trans , \bbeta^* + \mcA/\sqrt{n} \} ;\\ 
	&\M{s.t.} \quad \|\bGamma^{1/2}\a\| \leq h, \a \in\mathbb{R}^p, \b \in \mathbb{R}^q, \|\b\|\leq h, b_{\ell_k}=0, \|\mcA\|\leq h .
	\end{align*}
	We claim that for any $\epsilon > 0$, there exists a large enough $h$ such that 
	\begin{align}
	P\left\{ \underset{\substack{\|\Gamma^{1/2}\a\| = \|\b\| = \\ \|\mcA\|  = h}}{\M{inf}}  F_k^{(n)}(\what{\u}_k, \what{\v}_k, \what{\bbeta})  >  F_k^{(n)}(\u_k^*,\v_k^*,\bbeta^*)  \right\} \geq 1-\epsilon,  \label{the:consist12}
	\end{align}
	where $\what{\u}_k = \u_k^*+\a/\sqrt{n}, \,\, \what{\v}_k = \v_k^*+\b/\sqrt{n}, \,\,  \what{\bbeta} = \bbeta^*+\mcA/\sqrt{n}$, and the  offset term $\bO_k = \X \sum_{l=1}^{k-1}\what{\C}_l$. 
	For $\what{\C}_k = \what{\u}_k \what{\v}_k\trans$, define the natural parameter matrix estimate $\what{\bs\Theta}_k = \bO_k + \X\what{\C}_k + \Z\what{\bbeta}$. Similarly, define the  corresponding true natural parameter matrix $\bs\Theta_k^* = \bO_k + \X\C_k^* + \Z\bbeta^*$. Following \eqref{sec2:psifun}, for the $k$th step estimate, we formulate
	\begin{align}
	\Psi_k^{(n)}(\a,\b,\mcA)  & = F_k^{(n)}(\u_k^*+\a/\sqrt{n},\v_k^*+\b/\sqrt{n}, \bbeta^*+\mcA/\sqrt{n})  -  F_k^{(n)}(\u_k^*,\v_k^*,\bbeta^*) \notag \\
	& = F_k^{(n)}(\what{\Theta}_k) - F_k^{(n)}(\Theta_k^*) \notag \\
	& = T_1 + T_2 + T_3, \label{psifunk} 
	\end{align}
	where $T_1$, $T_2$, $T_3$ are obtained by  replacing $(\bO_1,\u_1^*,\v_1^*,\W_1,\lambda_1^{(n)})$ with $(\bO_k,\u_k^*,\v_k^*,\W_k,\lambda_k^{(n)})$ in equation \eqref{sec2:TTT}. Again, we write $T_1=T_{11}+T_{12}$, where
	\begin{align*}
	T_{11} & =\Tr(\E\trans \{\bs\Theta_k^* -\what{\bs\Theta}_k\} ) \\
	T_{12} &=\Tr(\{\bB^'(\bs\Theta_k^*) -\bB^'(\bs\Theta^*)\}\trans \{\what{\bs\Theta}_k - \bs\Theta_k^* \})+ {\gammaL \over 2} \|\what{\bs\Theta}_k - \bs\Theta_k^* \|_F^2.
	\end{align*}
	Following the proof  for  $k=1$, we conclude that $T_{11}$, $T_2$ and $T_3$ are of $O(h)$. 
	
	Again, to simplify $T_{12}$, we apply assumption \textbf{A4} 
	and conveniently write
	$$
	\Tr( \{\bB^'(\bs\Theta_k^*) -\bB^'(\bs\Theta^*)\}\trans \{\what{\bs\Theta}_k - \bs\Theta_k^*\} )\geq \gammaL\Tr(\{\bs\Theta_k^* -\bs\Theta^*\}\trans \{\what{\bs\Theta}_k - \bs\Theta_k^*\} )
	$$
	which results in 
	\begin{align*}
	T_{12} \geq   \gammaL (\Tr( \{\bs\Theta_k^* -\bs\Theta^* \}\trans\{\what{\bs\Theta}_k - \bs\Theta_k^*\} ) + {1 \over 2} \|\what{\bs\Theta}_k - \bs\Theta_k^* \|_F^2  ). 
	\end{align*}
	On replacing  $\bs\Theta_k^*, \bs\Theta^*$ and $\what{\bs\Theta}_k$ with their linear forms, we have 
	\begin{align*}
	T_{12} \geq   \gammaL (\Tr( \{ \X \sum_{i=1}^{k-1} (\what{\C}_i - \C_i^*)-\X \C_{-k}^*\}\trans & \{ {\Z\mcA \over \sqrt{n}} +  { \X \over \sqrt{n} } (\u_k^* \b\trans + \a\v_k\strans + {\a\b\trans \over \sqrt{n}} ) \} )\\ & + {1 \over 2n} \| \Z\mcA +  \X (\u_k^* \b\trans + \a\v_k\strans)\|_F^2 ),
	\end{align*}
	where $\C_{-k}^* = \sum_{l>k}^{r^*} \C_{l}^* $. Given  assumption \textbf{A2} and the $\sqrt{n}$ consistency of $\what{\C}_i$ for $i < k$, we have
	\begin{align*}
	&{1 \over 2n} \| \Z\mcA +  \X (\u_k^* \b\trans + \a\v_k\strans)\|_F^2 = {1 \over 2n} (\| \Z\mcA\|_F^2 +  \|\X (\u_k^* \b\trans + \a\v_k\strans)\|_F^2),\\
	&\Tr( \{-\X \C_{-k}^* \}\trans  { \X \over \sqrt{n} } (\u_k^* \b\trans + \a\v_k\strans + {\a\b\trans \over \sqrt{n}} )  )= -\sum_{l>k}^{r^*} \a\trans \bGamma_1\u_l^* \b\trans \v_l^* ,\\
	&\Tr( \{-\X \C_{-k}^*\}\trans {\Z\mcA \over \sqrt{n}})= 0 , \\
	&\Tr( \{\X \sum_{i=1}^{k-1} (\what{\C}_i - \C_i^*) \}\trans  {\Z\mcA \over \sqrt{n}}) = 0 ,\\
	&\Tr(\{ \X \sum_{i=1}^{k-1} (\what{\C}_i - \C_i^*)\}\trans  { \X \over \sqrt{n} } (\u_k^* \b\trans + \a\v_k\strans + {\a\b\trans \over \sqrt{n}} )  )= O(h) .
	\end{align*}	
	Hence,
	\begin{align*}
	T_{12} \geq   \gammaL ( {1 \over 2n}\| \Z\mcA\|_F^2 +  {1 \over 2n}\|\X (\u_k^* \b\trans + \a\v_k\strans)\|_F^2  -\sum_{l>k}^{r^*} \a\trans \bGamma_1\u_l^* \b\trans \v_l^* + O(h)  ).
	\end{align*}		
	The rest of the proof is similar to the case of $k=1$, where we prove the result by following  Theorem 5.2 of \citet{mishra2017sequential}. 	This completes the proof for the case of GOFAR(S). 
	
	\paragraph{CASE - GOFAR(P)} In proving the result for GOFAR(P), we mainly follow the steps of the proof for the case of GOFAR(S). In the $k$th step of GOFAR(P), we set the  offset term $\bO_k = \X \sum_{l \neq k}\widetilde{\C}_l$ and  define $\what{\bs\Theta}_k = \bO_k + \X\what{\C}_k + \Z\what{\bbeta}$ and $\bs\Theta_k^* = \bO_k + \X\C_k^* + \Z\bbeta^*$. The two approaches differ mainly in simplifying the term $T_{12}$, i.e., 
		\begin{align*}
	T_{12} \geq   \gammaL (\Tr( \{\X \sum_{i \neq k} (\widetilde{\C}_i - \C_i^*) \}\trans & \{  {\Z\mcA \over \sqrt{n}} +  { \X \over \sqrt{n} } (\u_k^* \b\trans + \a\v_k\strans + {\a\b\trans \over \sqrt{n}} ) \} )\\ & + {1 \over 2n} \| \Z\mcA +  \X (\u_k^* \b\trans + \a\v_k\strans)\|_F^2 ).
	\end{align*}
Given assumption \textbf{A2} and the  $\sqrt{n}$ consistency of $\widetilde{\C}_i$ for $i \neq  k$    (assumption \textbf{A5}), we have
	\begin{align*}
	&{1 \over 2n} \| \Z\mcA +  \X (\u_k^* \b\trans + \a\v_k\strans)\|_F^2 = {1 \over 2n} (\| \Z\mcA\|_F^2 +  \|\X (\u_k^* \b\trans + \a\v_k\strans)\|_F^2),\\
	&\Tr(\{ \X \sum_{i \neq k} (\widetilde{\C}_i - \C_i^*) \}\trans {\Z\mcA \over \sqrt{n}}) = 0 ,\\
	&\Tr( \{ \X \sum_{i \neq k} (\widetilde{\C}_i - \C_i^*)\}\trans  { \X \over \sqrt{n} } (\u_k^* \b\trans + \a\v_k\strans + {\a\b\trans \over \sqrt{n}} ))= O(h).
	\end{align*}	
	This results in 
	\begin{align*}
	T_{12} \geq   \gammaL ( {1 \over 2n}\| \Z\mcA\|_F^2 +  {1 \over 2n}\|\X (\u_k^* \b\trans + \a\v_k\strans)\|_F^2   + O(h)  ).
	\end{align*}		
The rest of the proof is similar to the case of $k=1$, where we prove the result by following  Theorem 5.2 of \citet{mishra2017sequential}. 	This completes the proof for the case of GOFAR(P).



\end{document}